\documentclass{statsoc}

\usepackage[a4paper, margin=0.95in]{geometry}
\usepackage{graphicx}
\usepackage{amsmath,amsfonts,amssymb, amsthm}
\usepackage{natbib}

\usepackage{adjustbox}
\usepackage{enumerate}
\usepackage{booktabs}
\usepackage{multirow}
\usepackage{bm}
\usepackage{xcolor}
\usepackage{microtype} 
\usepackage{import}
\usepackage{threeparttable}

\newtheorem{theorem}{Theorem}
\newtheorem{proposition}{Proposition}
\newtheorem{corollary}{Corollary}
\theoremstyle{definition}
\newtheorem{definition}{Definition}

\newtheorem{remark}{Remark}

\usepackage{url}

\usepackage{breakurl}
\usepackage[breaklinks]{hyperref}
\definecolor{mblue}{rgb}{0,0.4470,0.7410}
\hypersetup{colorlinks=true,linkcolor=mblue,urlcolor=mblue,citecolor=mblue}

\newcommand{\ob}[1]{{#1}}

\title[On the Reliability of Multiple Systems Estimation]{On the Reliability of Multiple Systems Estimation for the Quantification of Modern Slavery}
\author{Olivier Binette and Rebecca C. Steorts}
\address{Duke University, Durham, United States.}
\email{beka@stat.duke.edu}

\begin{document}

\begin{abstract}
The quantification of modern slavery has received increased attention recently as organizations have come together to produce global estimates, where multiple systems estimation (MSE) is often used to this end. Echoing a long-standing controversy, disagreements have re-surfaced regarding the underlying MSE assumptions, the robustness of MSE methodology, and the accuracy of MSE estimates in this application. Our goal is to help address and move past these controversies. To do so, we review MSE, its assumptions, and commonly used models for modern slavery applications. We introduce all of the publicly available modern slavery datasets in the literature, providing a reproducible analysis and highlighting current issues. Specifically, we utilize an internal consistency approach that constructs subsets of data for which ground truth is available, allowing us to evaluate the accuracy of MSE estimators. Next, we propose a characterization of the large sample bias of estimators as a function of misspecified assumptions. Then, we propose an alternative to traditional (e.g., bootstrap-based) assessments of reliability, which allows us to visualize trajectories of MSE estimates to illustrate the robustness of estimates. Finally, our complementary analyses are used to provide guidance regarding the application and reliability of MSE methodology.
\end{abstract}

\section{Introduction}
Modern slavery refers ``to situations of exploitation that a person cannot refuse or leave because of threats, violence, coercion, deception, and/or abuse of power'' \citep{ILO2017, ILO2017b}. This term encompasses issues of forced labor, forced sexual exploitation, and forced marriage \citep{ILO2017}. Individuals involved in the recruitment, harboring, and receipt/transportation of victims of such exploitation are referred to as human traffickers \citep{UnitedNations2000, UnitedNations2000b, Sigmon2008}.\footnote{The terms ``modern slavery,'' ``contemporary forms of slavery,'' and ``human trafficking'' are commonly considered synonymous. However, the terminology has been subject to variation and debate due to its history and due to various legal definitions \citep{Feingold2010, Chuang2014, Cockayne2015, Dottridge2017, Mende2019, Scarpa2020, Allain2017, Davidson2015, Piper2015, Bunke2016}. In this paper, we use the term \emph{modern slavery} as defined by the International Labour Organization and as commonly referred to throughout the literature (section~\ref{sec:data}).}

%\subsubsection{Quantifying modern slavery using multiple systems estimation}
The quantification of modern slavery has received increased attention as organizations have come together to produce global estimates \citep{WalkFree2013, Datta2013, ILO2017, Landman2020}. These efforts have {assisted} anti-slavery campaigns, {led to} increased public awareness, 
and {supported} newly implemented governmental policies. One major, recent goal is {determining the number of} victims of modern slavery. Specifically, in the United Kingdom (UK),  the Home Office estimated between $10,000$ and $13,000$ potential victims of modern slavery in 2013 using multiple systems estimation (MSE) \citep{Silverman2014, Bales2015}. Producing this estimate was part of the strategy leading to the UK Modern Slavery Act 2015 \citep{UKact2015}. Similar studies have been carried out in the Netherlands \citep{VanDijk2017}, in New Orleans \citep{Bales2019}, in the Western United States (U.S.) \citep{Farrel2019}, in Australia \citep{Lyneham2019}, as well as in Serbia, Ireland and Romania \citep{UNODC2018a, UNODC2018b, UNODC2018c}. These are reviewed in section \ref{sec:data}. 

MSE, reviewed in section \ref{sec:background}, is often the only available technique to estimate the prevalence of modern slavery. This is because victims of modern slavery can be out of the reach of traditional surveys. Instead of relying on representative samples, MSE uses case reports from multiple organizations, such as the police and non-governmental organizations, in order to estimate the population size. This approach promises to expose the scale of an issue which could otherwise be ignored.

The foundations of MSE dates back to at least the 17th century, when John Graunt used similar ideas to estimate London's population in 1661 \citep{hald2005history}. Laplace later formalized the technique to provide error bounds \citep{Laplace1820}, Quetelet advocated this approach for more efficient population census \citep{Quetelet1827, stigler1986history}, and these ideas were further developed in \cite{sekar1949method}. Today, MSE (also referred to as capture-recapture) is widespread in population ecology, in epidemiology, for official statistics and in the social sciences \citep{Bird2018, bohning2017capture}.

\subsection{Two Centuries of Controversy}
Despite its potential, MSE has faced harsh criticism for nearly two centuries. Quetelet abandoned MSE following concerns raised by de Keverberg on the soundness of underlying assumptions \citep{Quetelet1827, stigler1986history}. More recently, Cormack 
%--- a reference in the field of capture-recapture for population ecology  --- 
expressed deep concerns regarding the use of similar methods in epidemiological applications \citep{Cormack1999}. He stated that ``many of these studies give estimates which are not scientifically justified by the underlying data.'' This led to an energetic correspondence between Cormack and proponents of MSE \citep{Hook1999, Cormack1999b, Hook2000, Cormack2000}. In the context of MSE for the quantification of modern slavery, the issues raised by Cormack and the following disagreement was almost exactly repeated in \cite{Whitehead2019} and in the correspondence with ensued with the authors of key modern slavery studies \citep{Vincent2020, Whitehead2020, Vincent2020b}. {The disagreements involved the soundness of assumptions underlying MSE, the robustness of the procedures, and the accuracy of estimates in applications.}

%\textbf{These disagreements center around two key issues:
%\begin{enumerate}
%\item whether or not multiple systems estimation fundamentally differs from other statistical methodology in the assumptions it requires,
%\item  how the practical usefulness of estimates should be weighed against the need for scientifically accurate and robust estimates.
%\end{enumerate}
%}

\subsection{Our Contribution}
Our goal is to help address these controversies. \ob{This is challenging because, as stated by \cite{Silverman2020}, ‘‘no ‘ground truth’ is available to investigate the accuracy of any estimates.'' Instead, we address this issue using the following key observations:}
%
%However, while \cite{Silverman2020} stated that ‘‘no ‘ground truth’ is available to investigate the accuracy of any estimates,''  we address this issue using the following key observations:
\begin{enumerate}
    \item data with ground truth can be obtained from available datasets using the \textit{internal consistency approach} of \cite{hook2000accuracy, Hook2012};
    \item the convergence and bias of estimates (due to possibly misspecified assumptions) can be characterized in an asymptotic framework; and
    \item the reliability of estimates can be diagnosed using resampling techniques.
\end{enumerate}

Regarding (a), we use all publicly available data from past MSE studies on modern slavery \citep{Silverman2014, Bales2015, Bales2019, VanDijk2017, Farrel2019, Lyneham2019}. The internal consistency approach constructs subsets of this data for which ground truth is available, allowing us to evaluate the accuracy of MSE estimators. \ob{This approach was recommended in the discussion of \cite{Silverman2020} by \cite{Ridout2020}, but has not, to our knowledge, previously been carried out for modern slavery data. Other types of internal validation approaches have been suggested as a fruitful avenue for future research in discussions by \cite{Bohning2020} in \cite{Silverman2020}.}
Regarding (b), we introduce a novel characterization of the large sample bias of estimators as a function of misspecified assumptions. Using our results, we quantify the effect of individual heterogeneity on the bias of estimates. This shows the scale and direction of the bias that can be expected in reasonable practical situations.
Regarding (c), we propose an alternative to traditional (e.g., bootstrap-based) assessments of reliability. Our proposal is more easily interpretable and it involves fewer degrees of freedom in its specification. In practice, it allows us to put modern slavery MSE estimates into the perspective of a hypothetical \textit{trajectory} of estimates. These trajectories are visualized to showcase \ob{the robustness} of estimates to small changes in the data.
In addition to these three contributions, we provide a thorough review of MSE methodology and a comparison of real data estimates. Our analyses \ob{allow us to} illustrate obstacles within MSE, \ob{to} provide practical recommendations, and \ob{to} provide further directions for research, \ob{complementing previous work in this area \citep{Silverman2020, Far2021}.}

%Most importantly, our points and proposals are complementary in nature. {For point (1)}, we provide an evaluation of accuracy which applies to \textit{subsets} of the data and \textit{conditionally} on a set of assumptions. {For point (2)}, we showcase the consequences of underlying assumptions. {For point (3),} we give insights on the full data rather than only subsets.

%To illustrate the effectiveness of our proposals, we consider all past modern slavery studies for which data is publicly available \ref{sec:data_and_estimators}. Finally, we provide guidance regarding how our methods can be used moving forward to users in the field. 

%For this empirical evaluation, we consider all past modern slavery studies using multiple systems estimation which we are aware of and for which data is publicly available. 

%The datasets and estimators which we consider are summarized in Section \ref{sec:data}.
%Given our analyses and results, we conclude in section \ref{sec:discussion} with a discussion of the main statistical challenges which come in the way of the reliability of multiple systems estimation studies. 

\subsection{Organization of the Paper}

The rest of the paper is organized as follows. Section~\ref{sec:background} reviews the MSE literature.
%, where we emphasize underlying assumptions and the influence of such assumptions on the accuracy of estimates.
%multiple systems estimation and the estimators which are later compared. Particular emphasis is given to underlying assumptions and to the influence of these assumptions on the accuracy of estimates.
Section~\ref{sec:data_and_estimators} describes data from past modern slavery studies, providing comparisons, performing a sensitivity analysis, and discussing MCMC convergence issues with one of the Bayesian models. Section~\ref{sec:internal_consistency_analysis} evaluates the performance of estimators when ground truth is available using the internal consistency approach. Section \ref{sec:large-sample-bias} provides results regarding the bias of estimates, including the consequences of individual heterogeneity. Section~\ref{sec:visual_assessment} proposes a visual diagnostic of estimator robustness and showcases its use on modern slavery data. Finally, we discuss main takeaways in section~\ref{sec:discussion}.

\section{Background} \label{sec:background}

{This section reviews} the general framework of MSE. We describe the fundamental idea as reflected in the Lincoln-Peterson estimator in section~\ref{sec:MSE_intro} and we introduce the general model with multiple lists in section~\ref{sec:general_framework}. 
Section~\ref{sec:assumptions} discusses the assumptions of MSE and their relationship to the existence of consistent population size estimators. We then review {log-linear, decomposable graphical, and latent class models} in section~\ref{sec:modeling-approaches}.

%We conclude in section \ref{sec:challenges} with a discussion of challenges which are specific to multiple systems estimation.

\subsection{Introduction to Multiple Systems Estimation}\label{sec:MSE_intro}

MSE is a technique used to estimate the total size of a population. It relies on multiple samples from the population, referred to as \textit{lists}, which have been collected by different organizations.
%
%The total number of \textit{observed} victims is obtained by combining the records from these sources and removing duplicates. On the other hand, the number of \textit{unobserved} victims is inferred through the patterns of overlap between the sources. Together, the number of observed and unobserved victims add up to the total population size. 
%
The basic principles of the approach are most intuitively explained in the context of two lists. Here two organizations (e.g. the police and a non-governmental organization) record their contact with individuals from the population of interest (such as potential victims of modern slavery).
A large overlap between the two lists may indicate that a large portion of the total population was observed. A small overlap may indicate that a larger portion was unobserved. This is justified if the two lists are independent, meaning that the probability that individual appears on one of the lists does not depend on whether an individual appears on the other list.

To formalize this idea, let $n_1$ be the number of potential victims appearing on a first list, let $n_2$ be the number on the second list, and let $m$ be the number of potential victims appearing on both. The classical Lincoln-Petersen estimator \citep{Lincoln1930, Petersen1895} of the total population size is  defined as
\begin{equation}\label{eq:lincoln_petersen_estimator}
    \hat N = \frac{n_1 n_2}{m}.
\end{equation}
The estimator is nearly unbiased if the two lists are independent (see \cite{bishop2007discrete, Chao2008} for a full account of the properties of the Lincoln-Petersen estimator, extensions, and applications).
This was first used to estimate the number of fish in closed reservoirs and extensions led to MSE methodology that is widely used in population ecology \citep{cormack1968statistics, Seber1982, Seber1986, Seber1992, Amstrup2005}. These extensions have also been adopted in epidemiology \citep{Wittes1974, Yip1995, Yip1995b, Chao2001} and to inform public policy \citep{Bird2018}, as well as in human rights applications \citep{Lum2013, ManriqueVallier2013}.

%In practice, the assumption of independence between samples if often unrealistic. This can be partly addressed by considering multiple samples from the population and by estimating interactions between samples using the data at hand. Methods for this are described in section {[XX]}.
%This is described in full detail in section \ref{sec:MSE}, where we introduce the statistical model for MSE and its underlying assumptions.

\subsection{The General Framework of Multiple Systems Estimation}\label{sec:general_framework}

Generally, more than two lists are used in MSE studies. This can provide greater coverage of the population of interest and  allows the estimation of certain interactions between lists. In this section, we review the general MSE model used for these purposes.

%This is used in section {[XX]} to showcase the particular challenges of multiple systems estimation applied to modern slavery. 
%This section introduces the standard MSE model and its underlying assumptions.

As previously stated, MSE is used to estimate the unknown size $N$ of a population when a complete enumeration is not possible. Instead, $L \geq 2$ lists of observed cases are used to draw inference. Each list records a small fraction of the population. The total number of \textit{observed} individuals, $n_{\text{obs}}$, is obtained by combining lists and removing any duplicate individuals. The number of \textit{unobserved} individuals,  $n_{\bm{0}}$, is estimated using MSE from the patterns of overlap between the lists. Together, the number of observed individuals ($n_{\text{obs}}$) and the number of \textit{unobserved} individuals ($n_{\bm{0}}$) account for the entire population, $N = n_{\text{obs}} + n_{\bm{0}}$.

For each individual $i = 1,2,3,\dots, N$ in the entire population, we observe a \textit{list inclusion pattern} $W_i = (W_{i,1}, W_{i,2}, \dots, W_{i,L}) \in \{0,1\}^L$ where $W_{i,j} = 1$ if individual $i$ appears on list $j$ and $W_{i,j} = 0$ otherwise. If $W_i = (0,0, \dots, 0) = \bm{0}$, then the $i$th individual is unobserved. 
The \textit{observed data} are the counts
\begin{equation}\label{eq:def_counts}
    n_x = \sum_{i=1}^N \mathbb{I}(W_i = x),\quad x \in \{0,1\}^L \backslash \{\bm{0}\}.
\end{equation}
This represents the counts of individuals with given non-zero inclusion patterns. The set of all observed counts is denoted by $(n_x)_{x \not = \bm{0}}$ and the number of observed individuals is $n_{\text{obs}} = \sum_{x \not = \bm{0}} n_x$. 

%{[TODO: Talk about the use of categorical covariates and why we can ignore it almost WLOG in this paper. N.B. there are two standard ways to account for covariates: stratification and inclusion of the covariates in the log-linear model.]}

For example, in the context of the United Kingdom (UK) study \citep{Silverman2014, Bales2015}, the lists correspond to organizations coming into contact with potential victims of modern slavery. The organizations are local authorities (LA), the police force (PF), the national crime agency (NCA), governmental organizations (GO), non-governmental organizations (NG), and the general public (GP). Each organization records identifying information for each case it comes into contact with. The resulting lists are then \textit{matched} together using record linkage or de-duplication \citep{christen_data_2012, Christophides2019, Binette2020}, which allows one to identify duplicate records from multiple lists. The observed data, containing all observed overlap counts, is reproduced in table \ref{table:UKdata}, where the PF and NCA lists have been combined following \cite{Silverman2014, Bales2015}. In table \ref{table:UKdata}, columns under ``Cases observed once'' represent the number of victims only observed in the list marked by an ``$\times$'' underneath. The other columns represent the amount of overlap between the lists marked by ``$\times$'' underneath. For example, 54 potential victims have only been reported by the LA list, that 463 potential victims have only been reported by the NG list, and that 15 potential victims have been reported by both LA and NG but not by any of the other lists. One victim has been observed on LA, NG, PFNCA and GO, but not on GP (rightmost column). In total, 2744 distinct potential victims have been identified in this dataset.

%\textcolor{red}{List the other studies and organize these tables in an Appendix with their tables.}

\begin{table}
    \caption{\label{table:UKdata} Counts of potential victims of modern slavery in the UK, disaggregated by the lists in which potential victims appear. This data was reported in \cite{Silverman2014}.}
    \centering
\begin{tabular}{r|ccccc|ccccccccc|cccc}
     Total& \multicolumn{5}{c}{Cases observed once} & \multicolumn{9}{c}{Cases observed twice} & \multicolumn{4}{c}{3+ times} \\
     2744 & 54 & 463 & 995 & 695 & 316 & 15 & 19 & 3 & 62 & 19 & 1 & 76 & 11 & 8 & 1 & 1 & 4 & 1 \\\midrule
    LA & $\times$ &  &  &  &  & $\times$ & $\times$ & $\times$ &  &  &  &  &  &  & $\times$ & $\times$ &  & $\times$ \\ 
    NG &  & $\times$ &  &  &  & $\times$ &  &  & $\times$ & $\times$ & $\times$ &  &  &  & $\times$ & $\times$ & $\times$ & $\times$ \\ 
    PFNCA &  &  & $\times$ &  &  &  & $\times$ &  & $\times$ &  &  & $\times$ & $\times$ &  & $\times$ &  & $\times$ & $\times$ \\ 
    GO &  &  &  & $\times$ &  &  &  & $\times$ &  & $\times$ &  & $\times$ &  & $\times$ &  & $\times$ & $\times$ & $\times$ \\ 
    GP &  &  &  &  & $\times$ &  &  &  &  &  & $\times$ &  & $\times$ & $\times$ &  &  &  &  \\ 
    \bottomrule
\end{tabular}
\end{table}

\subsection{Assumptions and Consistency of Estimators}\label{sec:assumptions}

%% N.B. all assumptions deal with the inclusion patterns.
{This section reviews} the three assumptions of the standard MSE model {regarding list inclusion patterns}. The first two are standard, stating the data is independently and identically distributed. The third is an identifiability assumption which we show in Proposition \ref{prop:1} is necessary to the existence of consistent population size estimators.

%\textcolor{red}{OB: State your contributions here to the literature. It's important that you get practice doing this and showcasing this so they are not hidden. Forward link to any important results.} {[Do you mean to forward-reference Proposition 1, or to forward-reference other results?]} In addition, we comment on how realistic such assumptions are in practice. 
%\textcolor{red}{{RCS: This is awkward in construction: The first two assumptions, of independence and identical distribution of the data, are formally stated below:}}
{The first two assumptions are formally stated below:}

%The first two assumptions of the standard MSE model are assumptions of independence and identical distribution of the data:

\begin{description}
    \item[A1] The list inclusion patterns $W_i$, $i=1,2,3,\dots, N$ are \textit{independent} from one another. That is, the lists on which individuals appear or do not appear has no influence on the inclusion patterns of other individuals \ob{ --- this is independence across individuals.}
    \item[A2] The list inclusion patterns $W_i$, $i=1,2,3,\dots, N$ are \textit{identically distributed} with
    \begin{equation}\label{eq:def_px}
        p_x = \mathbb{P}(W_i = x) > 0,\quad x \in \{0,1\}^L. 
    \end{equation}
    That is, for any set of lists represented by an inclusion pattern $x \in \{0,1\}^L$, all individuals are equally likely to have $x$ as an inclusion pattern.
\end{description}

\textbf{A1} and \textbf{A2} are classical assumptions in the MSE literature \ob{and we point the reader to \citep{Lum2013} for a practical discussion of MSE assumptions.} While assumption \textbf{A1} may not always hold, we expect that it has a non-negligible effect for large populations.
Assumption \textbf{A2} is also less stringent than it appears. It only requires the data to be \textit{marginally} identically distributed. That is, suppose that the inclusion pattern probability of an individual $i$ is affected by an unobserved variable $\lambda_i,$ which represents the type of crime involved or socio-demographic characteristics of the individual victim. As long as $\mathbb{E}_{\lambda_i}\left[\mathbb{P}(W_i = x \mid \lambda_i)\right] = p_x$ is constant and does not depend on $i$, then assumption \textbf{A2} is satisfied.
{Such models, where inclusion probabilities depend on latent individual characteristics, are referred to as individual \textit{heterogeneity} models \citep{Otis1978}. These are relevant to modern slavery applications given the heterogeneity in the population which may impact list inclusion probabilities. Throughout the paper, we assume that \textbf{A1} and \textbf{A2} hold. Heterogeneity models are specifically considered in section~\ref{sec:bias_heterogeneity}.}

%While in some cases potential victims may be related to one another or may be part of a group (\textbf{A1}), we do not expect such strong relationships between potential victims to be present at a large scale. Assumption \textbf{A2} is also less stringent than it appears. It only requires the data to be \textit{marginally} identically distributed. That is, suppose that the inclusion pattern probability of an individual $i$ is affected by an unobserved variable $\lambda_i,$ which represents the type of crime involved or socio-demographic characteristics of the individual victim. As long as $\mathbb{E}_{\lambda_i}\left[\mathbb{P}(W_i = x \mid \lambda_i)\right] = p_x$ is constant and does not depend on $i$, then assumption \textbf{A2} is satisfied. Therefore, we will assume that these assumptions hold throughout the paper.

%\textcolor{red}{Given that assumptions \textbf{A1} and \textbf{A2} are standard in the literature, we will also assume these to hold throughout the paper.}

{Here, we introduce a decomposition of the data likelihood due to \cite{Fienberg1972}.}
Under assumptions \textbf{A1} and \textbf{A2}, the observed data $(n_x)_{x \not = \bm{0}}$ is distributed as
\begin{align}
    n_{\text{obs}} &\sim \text{binomial}(1-p_{\bm{0}}, N),\label{eq:gen_model_n}\\
    (n_x)_{x \not = \bm{0}}\mid n_{\text{obs}} &\sim \text{multinomial}\left((q_x)_{x \not = \bm{0}}; n_{\text{obs}}\right),\label{eq:gen_model_counts}
\end{align}
where $$q_x = p_x / (1-p_{\bm{0}}) = \mathbb{P}(W_i = x \mid W_i \not = \bm{0}), \quad x \not = \bm{0}$$ is the conditional probability of the inclusion pattern $x$ given the individual being observed. {In some cases, a Poisson likelihood is used as an approximation to the multinomial model \eqref{eq:gen_model_n} and \eqref{eq:gen_model_counts} \citep{Cormack1989}. This is the case with the Poisson log-linear modeling approach reviewed in section \ref{sec:estimators}. The Poisson likelihood is used for convenience and does not make any important difference on the model, its assumptions, or resulting estimates.}

While the population size $N$ is identifiable in the standard MSE model {given by \eqref{eq:gen_model_n} and \eqref{eq:gen_model_counts}} \citep{Farcomeni2012}, assumptions \textbf{A1} and \textbf{A2} are not sufficient by themselves to obtain meaningful population size estimates. Indeed, it follows from \eqref{eq:gen_model_n} that the observed data provides information about $N$ only through $n_{\text{obs}}$ and $p_{\bm{0}}$. In section \ref{sec:large-sample-bias}, we formally define the notion of \textit{consistent} population size estimators to formalize what can be learned from the data. Roughly, a population size estimator $\hat N$ is consistent for a statistical model $\Theta \subset \left\{(p_x)_{x \not = \bm{0}}\,:\, \sum_x p_x = 1\right\}$ if and only if it converges to the true population size $N$ in large samples whenever the probabilities $p_x$ in \eqref{eq:def_px} are part of $\Theta$. The minimal requirement for the existence of consistent population size estimators is given in Proposition \ref{prop:1} below. 
 
\begin{proposition}\label{prop:1}
    If a model $\Theta$ admits a consistent population size estimator, then there exists a function $f$ such that
    $
        p_{\bm{0}} = f((q_x)_{x \not = \bm{0}})
    $
    for all $(p_x)_{x \in \{0,1\}^L} \in \Theta,$ where $q_x = \frac{p_x}{1-p_{\bm{0}}}$.
\end{proposition}

{The proof is in Appendix \ref{sec:proof_proposition}.} \\

{Proposition \ref{prop:1}} motivates {assumption \textbf{A3} which,} given some function $f$, restricts the set of probabilities $(p_x)$ under consideration.

\begin{description}
    \item[\textbf{A3}] There exists a function $f$ such that $$p_{\bm{0}} = f((q_x)_{x \not = \bm{0}}),\quad  \text{where} \quad q_x = \frac{p_x}{1-p_{\bm{0}}}.$$ That is, the unobserved probability $p_{\bm{0}}$ is the deterministic function $f$ of the observed data distribution through the cell probabilities $(q_x)_{x \not = \bm{0}}$.
\end{description}

We refer to this as the \textit{identifying assumption}, which formalizes a condition of \cite{Link2003}. That is, \textbf{A3} is equivalent to stating that no two different distributions $(p_x)_{x \in \{0,1\}^L}$ lead to the same zero-truncated distribution $(q_x)_{x \not = \bm{0}}$. Our assumption \textbf{A3} is equivalent to Definition 1 in \cite{AleshinGuendel2020} and Definition 2 in \cite{aleshin2021revisiting}. It relates part \eqref{eq:gen_model_counts} of the observed data distribution to the probability of an individual being unobserved. As such, it can be interpreted as specifying the \textit{missing data mechanism} at play in MSE. It can also be understood as an \textit{extrapolation formula} \citep{Manriquevallier2019}, which allows one to go from the observed data distribution $(q_x)_{x\not = \bm{0}}$ to the unobserved probability $p_{\bm{0}}$. Then through \eqref{eq:gen_model_n}, one can infer the total population size.  

An example of an assumption of the form \textbf{A3} is given in section \ref{sec:log_linear_model} in the context of log-linear modeling.

%The fundamental role of this assumption is discussed in detail in \cite{aleshin2021revisiting}.

\begin{remark}
    The observed data provides no information regarding assumption \textbf{A3}. In fact, as pointed out by \cite{Manriquevallier2019}, ``the way in which the probability $p_{\bm{0}}$ relates to the rest of $p_{x}$, $x \not = \bm{0}$, can neither be learned from data nor tested.''
Turning the choice of $f$ in \textbf{A3}, in practice, this is taken as the consequence of simpler and more easily interpretable assumptions. Typically, $f$ is chosen as the consequence of one of the particular modeling approaches described next section \ref{sec:modeling-approaches}.
\end{remark}

%The observed data provides no information regarding assumption \textbf{A3}. More specifically,  \cite{Manriquevallier2019} commented that ``the way in which the probability $p_{\bm{0}}$ relates to the rest of $p_{x}$, $x \not = \bm{0}$, can neither be learned from data nor tested.''
%
%In practice, the choice of $f$ in \textbf{A3} is taken as the consequence of simpler and more easily interpretable assumptions, or as part of the particular modeling approaches described next in section \ref{sec:modeling-approaches}. 

%The most common is the assumption of no full-way interaction term in the log-linear representation of the model, or the assumption of no three-way interaction terms in hierarchical log-linear models. This is further discussed in section {[XX]}.

%\begin{remark}
% In a Bayesian framework, assumption \textbf{A3} may be modified to specify a prior on $f$. This is discussed in section {[XX]}.
%\end{remark}

\subsection{Multiple Systems Estimation Methods}\label{sec:modeling-approaches}

In this section, we review MSE models widely used in the literature. 
These allow one to estimate the probabilities $p_x$ defined in \eqref{eq:def_px} and to specify
%for the specification of 
a function $f$ in \textbf{A3}, that provides via \eqref{eq:gen_model_n} a population size estimate. Crucially, we only consider models that do not inherently require covariate-level data as motivated by the applications in section \ref{sec:data}.  First, we review log-linear models \citep{Fienberg1972, Cormack1989}. Second, we review graphical models \citep{Madigan1995, Madigan1997}, which are a special case of log-linear models. Third, we review a family of latent class models \citep{ManriqueVallier2008, ManriqueVallier2016, AleshinGuendel2020}.

\subsubsection{Log-Linear Models}\label{sec:log_linear_model}

\cite{Fienberg1972} introduced the use of log-linear models for MSE, which provide the basis for many applications \citep{Yip1995, baillargeon2007rcapture}. Furthermore, the use of these models have been previously considered in the use of modern slavery studies (section \ref{sec:data_and_estimators}) \citep{Silverman2014, Bales2015, Bales2019, Silverman2020, Chan2020, VanDijk2017, Farrel2019, Lyneham2019}. 
In this section, we briefly review this literature, referring the reader to \citep{bishop2007discrete, Yip1995} for further details.

Log-linear models provide an interpretable re-parameterization of the model parameters $p_x$ and $N$ defined in \eqref{eq:def_px}. The log-linear model consists of an intercept term $\mu$, a main list effects term $\alpha_i$, two-way interaction terms $\beta_{i,j}$ (for lists $i\not = j$), and higher order interaction terms including a full-way interaction term $\gamma$. Thus, the log-linear parameterization of $N$ and $(p_x)_{x \in \{0,1\}^L}$ is given by
\begin{equation}\label{eq:log-linear-parameterization}
    \log (N p_x) = \mu + \sum_{i} x_i \alpha_i + \sum_{i\not = j} x_i x_j \beta_{i,j} + \cdots + x_1x_2\cdots x_L \gamma, \quad x \in \{0,1\}^L.
\end{equation}
The full-way interaction term $\gamma$ can be expressed as 
\begin{equation}\label{eq:def_gamma_sum}
    \gamma = \sum_{x\in \{0,1\}^L} (-1)^{\lvert x \rvert + 1} \log(p_x)
\end{equation}
where $\lvert x \rvert = \sum_i x_i$.

We now discuss the parameter interpretation in (\ref{eq:log-linear-parameterization}). If all parameters other than $\mu$ and the $\alpha_i$ are zero, this corresponds to independent lists. The probability that an individual appears in list $i$ is given by ${e^{\alpha_i}}/({1+e^{\alpha_i}})$. If two-way interaction terms are added, then $\beta_{i,j}$ represents how the odds of inclusion to non-inclusion on the $i$th list, conditionally on all other variables, changes depending on whether or not the individual appears on the $j$th list. For non-overlapping lists, setting $\beta_{i,j} =-\infty$ states that an individual cannot appear on both lists $i$ and $j$. Higher order interaction terms can be similarly interpreted in terms of log-odds changes.

One of the main advantages of the log-linear parameterization is the resulting model hierarchy. The independence model, with intercept $\mu$ and main effects $\alpha_i$, is obtained by setting all interaction terms to zero. More complex models can be obtained by adding interaction terms, allowing for dependencies between lists. Typically, simpler models are favored, with interaction terms only added to the extent that data provides evidence for them.  This has been called ``betting on sparsity'' \citep{friedman2001elements}.  For instance, all log-linear models can be fitted to the data using maximum likelihood estimation, and a single model can be selected based on Akaike's Information Criteria or through other criteria \citep{Chao2001}. \cite{regal1991effects, Yip1995} review other considerations involved in the selection of log-linear models and the reporting of corresponding estimates.

Crucially, the assumption of no full-way interaction in log-linear models, which corresponds to setting $\gamma = 0$ in \eqref{eq:log-linear-parameterization}, induces an identifying assumption of the form \textbf{A3}. The following assumption is typically made within the context of log-linear models \citep{Fienberg1972}:
\begin{description}
    \item[\textbf{A3.1}] The full-way interaction term $\gamma$ in the log-linear parameterization \eqref{eq:log-linear-parameterization} is zero. Equivalently, the probabilities $(p_x)$ defined in \eqref{eq:def_px}, with $q_x = p_x/(1-p_{\bm{0}})$, satisfy the relationship
    \begin{equation}\label{eq:assumptionA3.1}
        \log \frac{p_{\bm{0}}}{1-p_{\bm{0}}} = \sum_{x \not = \bm{0}} (-1)^{\lvert x \rvert} \log q_x.
    \end{equation}
\end{description}
This states that the $(L-1)$-way interaction term for individuals appearing on list $L$ is the same as the $(L-1)$-way interaction term for individuals not appearing on list~$L$. 

\begin{remark}
{As discussed in section \ref{sec:assumptions}, assumption \textbf{A3.1} is untestable and is made in order to obtain consistent population size estimators. \cite{Yip1995} states that this assumption ``is more likely to be approximately correct than assumptions about the absence of lower-order interactions.'' 
In the case of only two lists, \textbf{A3.1} is satisfied if the two list inclusion indicators $W_1$ and $W_2$ are uncorrelated or independent. With more than two lists, if one list is independent of another given the rest, then \textbf{A3.1} is also satisfied. }
\end{remark}
%

%In section \ref{sec:large-sample-bias}, we characterize the asymptotic bias of estimates when assumption \textbf{A3.1} is misspecified, such as in the presence of individual heterogeneity. \textcolor{red}{perhaps make it more clear that this is a new contribution so it's not missed.}

%Furthermore, the complete interpretation of \textbf{A3.1} can be given in a missing data framework. That is, \textbf{A3.1} represents a characteristics of the data which must be shared by both the individuals appearing on a given list $j$ and the individuals which do not appear on list $j$. This characteristic which must be shared is the full-way interaction term of the submodels induced by conditioning on inclusion or exclusion from the $j$th list. Assumption \textbf{A3.1} states that the full-way interaction terms of the induced submodels must be equal. In practice therefore, we can first restrict our attention to the cases appearing on the $j$th list (for which there is no missing count) in order to estimate the full-way interaction term of the corresponding submodel. Then, for the cases not appearing on the $j$th list (for which we're missing the number of unobserved individuals) we can use, through assumption \textbf{A3.1}, this same full-way interaction term as to extrapolate to individuals which do not appear on the $j$th list. All multiple systems estimation models based on assumptions \textbf{A1}-\textbf{A3.1} implicitely proceed in this way.

\subsubsection{Decomposable Graphical Models} \label{sec:decomposable-graphical-models}
{We now review} decomposable graphical models \citep{Darroch1980} with {hyper-Dirichlet} priors \citep{Dawid1993}, which were proposed by \cite{York1992, Madigan1995, Madigan1997} for MSE. First, we describe (undirected) graphical models, which are special cases of log-linear models with an intuitive interpretation of conditional dependencies. Second, we review decomposable graphical models with hyper-Dirichlet priors, which are mainly used for computational convenience as they provide a conjugate family for which population size estimates can be derived in closed form. 

\paragraph{Graphical Models}
Graphical models are statistical models where the dependency between variables is characterized by an interpretable graph. Nodes in the graph represent variables, and two variables are linked together if they exhibit certain conditional dependencies. In the context of MSE, the graph describes the dependencies and conditional independencies between lists. There is one node for each list; it represents the variable indicating whether or not a given individual appears on this list. The graph has \textbf{two equivalent properties}, known as the local Markov property or the conditional independence graph, which we define below.

\begin{description}
    \item[Local Markov property:] The conditional distribution of any variable only depends on other variables through its immediate neighbors in the graph.
    \item[Conditional independence property:] If $A$ and $B$ are two sets of vertices separated by another set $S$ in the graph, then the variables corresponding to $A$ and $B$ are conditionally independent given $S$.
\end{description}

Figure \ref{fig:graphs_examples} presents examples of conditional independence graphs for the list considered in the UK modern slavery study of \cite{Silverman2014, Bales2015}.

\begin{figure}[h]
    \centering
    \includegraphics[width=0.25\linewidth]{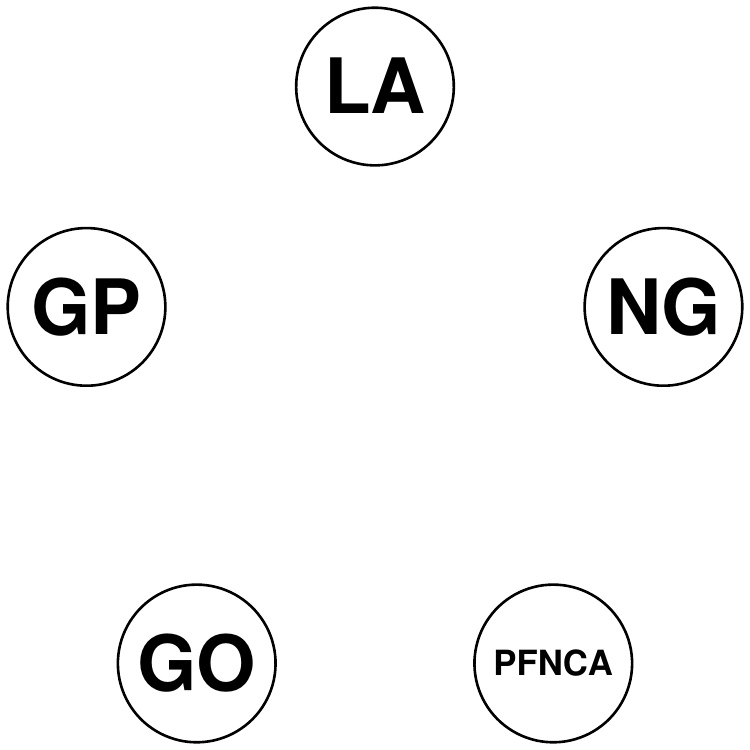}\hfill
    \includegraphics[width=0.25\linewidth]{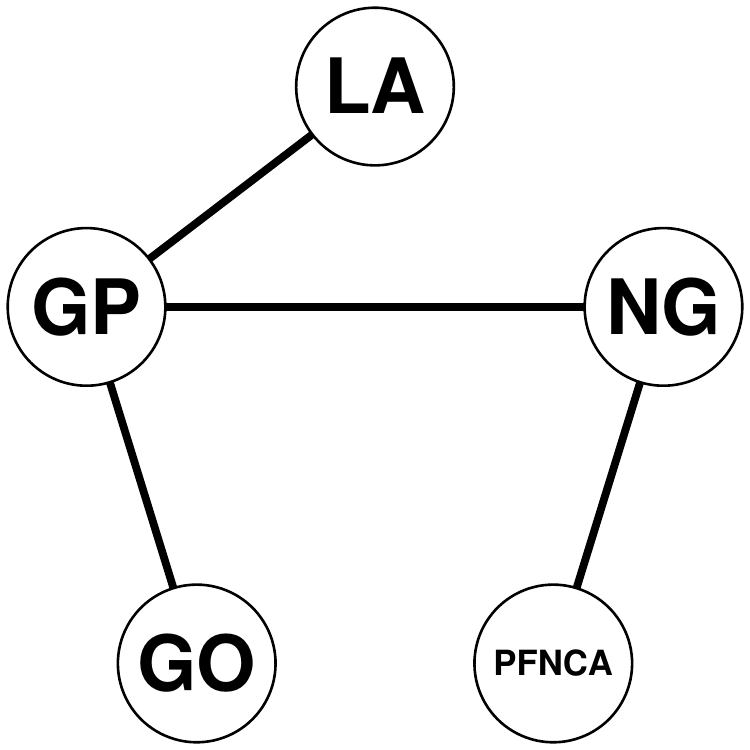}\hfill
    \includegraphics[width=0.25\linewidth]{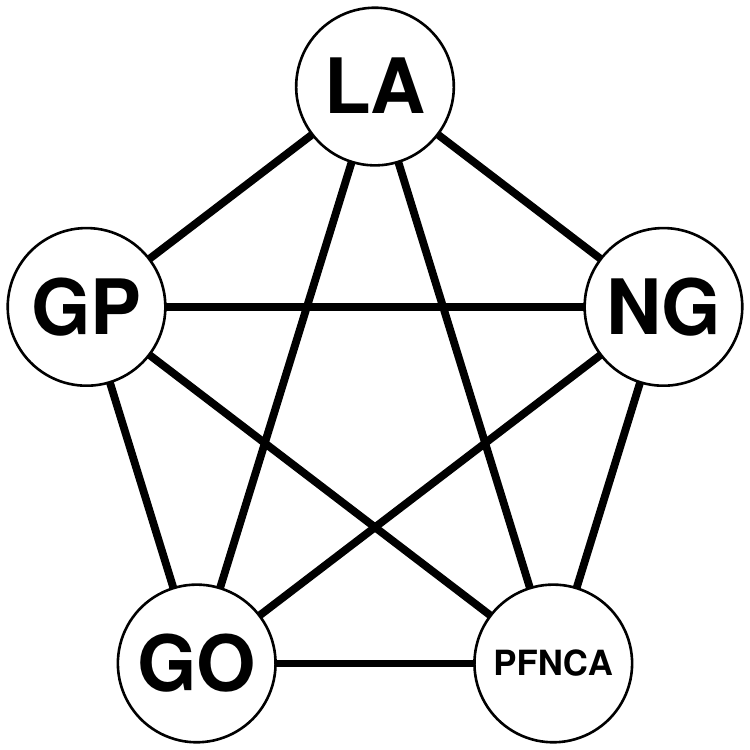}
    \caption{Examples of independence graphs on the 5 lists considered in the UK study of \cite{Silverman2014, Bales2015}. \textbf{Left}: Full independence between the lists. \textbf{Middle}: Various dependencies and conditional independencies. For example, NG is conditionally independent of GO and LA given GP. \textbf{Right}: Full dependency model with no non-trivial conditional independences.}
    \label{fig:graphs_examples}
\end{figure}

\paragraph{Decomposable Graphical Models and the Hyper-Dirichlet Prior}
%\textcolor{red}{TODO: This section is a bit difficult to get through. Go through it again and make sure everything is clearly defined and easy for someone who doesn't know this literature like you did with log-linear models.}\\
%\textcolor{red}{cycle is not defined nor clique. This is important to define and make more clear. Perhaps a picture here could be helpful.}
In this section, we review other important terminology used in graphical models. \textit{Decomposable} graphical models refer to graphical models for which the graph $G$ is  \textit{chordal} --- every cycle in the graph is part of a clique.\footnote{{A \textit{cycle} is a set of vertices which are connected in a closed chain. A \textit{clique} is a set of vertices which are all interconnected.}}
From a statistical point of view, decomposable graphical models have an important advantage in that every distribution on a decomposable graph is uniquely characterized by the set of marginal distributions over its cliques. {That is, no two different distributions on a decomposable graph can have the same marginal distributions over the cliques of the graph.} Furthermore, any set of \textit{pairwise consistent} marginal distributions over the cliques of a decomposable graph uniquely determines a joint distribution. Here, distributions on two sets of vertices $A$ and $B$ are said to be \emph{consistent} if they have the same marginal on $A \cap B$. {In other words, probability distributions over decomposable graphs can be specified through clique marginals which are pairwise consistent.}

\cite{Dawid1993} exploited these properties to define the \textit{hyper-Dirichlet prior}, a prior distribution for decomposable graphical models. This prior distribution is 
easily specified through Dirichlet clique marginals, allowing for tractable posterior inference.

\paragraph{MSE through Bayesian Model Averaging of Decomposable Graphical Models}

In the context of MSE and of \cite{Madigan1997}, the hyper-Dirichlet prior on decomposable graphical models is parameterized by a set of ``prior counts'' $\alpha_x$ for $x \in \{0,1\}^L$ \citep{Sadinle2018}. For a given decomposable graph, the marginalization of these counts over the cliques of the graph yields the corresponding parameters of the marginal Dirichlet priors. Given a prior on the population size, we obtain a tractable formula for the posterior distribution (see \cite{Madigan1997, Sadinle2018}). Uncertainty regarding the structure of the decomposable graph is incorporated using a prior on the set of all possible graphs, known as \emph{Bayesian model averaging} \citep{gelman2013bayesian}.

Excluding the complete graph, all graphical models are special cases of log-linear models with no full-way interactions \citep{Darroch1980}.
%This is a consequence of the Hammersley-Clifford Theorem [CITE]. 
As such, graphical modeling relies on assumption \textbf{A3.1} of no full-way interaction.

\subsubsection{Latent Class Models}
{Finally, we review} the latent class model of \cite{ManriqueVallier2016}, which is motivated by the modeling of latent individual heterogeneity.

The latent class model decomposes the probabilities $p_x$ defined in \eqref{eq:def_px} through a mixture representation. That is, the model takes the form
\begin{equation}\label{eq:latent_class_model}
p_x = \sum_{k=1}^K w_k \prod_{j=1}^{L} \lambda_{k,j}^{x_j}(1-\lambda_{k,j})^{1-x_j}
\end{equation}
where $w_k$ is the class weight and where $\lambda_{k,j}$ is the $j$th list inclusion probability for class $k$. Any distribution $(p_x)_{x \in \{0,1\}^L}$ can be decomposed through $(3)$ with $K = 2^{L-1}$ \citep{Johndrow2017}.
In terms of the list inclusion variables, this model represents independence between lists conditionally on the latent class to which belongs an individual.

\cite{AleshinGuendel2020, aleshin2021revisiting} showed model \eqref{eq:latent_class_model} induces an assumption of the form \textbf{A3} if and only if $2K \leq L$. When $2K > L$, it follows from \cite{AleshinGuendel2020} and Proposition \ref{prop:1} that no consistent population size estimator results from \eqref{eq:latent_class_model} with unrestricted values of $w_k$ and $\lambda_{k,j}$. 
In other words, given a prior on the latent class model, the resulting population size posterior distribution is generally inconsistent --- the posterior distribution does not converge to the true population size. Instead, such an approach quantifies uncertainty through the prior relationship between $p_{\bm{0}}$ and the other model probabilities. This is valid as long as the mixture of independence model is appropriate for the data, and as long as the prior specification properly captures our uncertainty regarding the distribution of the latent classes and the list inclusion probabilities. The prior specification of \cite{ManriqueVallier2016} for \eqref{eq:latent_class_model} is discussed in Section \ref{sec:estimators}.

\section{Data and Population Size Estimators Under Consideration}\label{sec:data_and_estimators}

{This section introduces the datasets and estimators which we consider throughout the paper. Section \ref{sec:data} reviews the datasets and section \ref{sec:estimators} introduces the estimators based on the models of section \ref{sec:modeling-approaches}. In section \ref{sec:comparison_estimates}, we then showcase how the corresponding estimates compare among themselves and in relation to published estimates. Finally in section \ref{sec:sensitivity_convergence}, we discuss challenges associated with these approaches, notably the sensitivities to choices of hyperparameters and convergence issues.}

\subsection{Data From Past Modern Slavery Studies}\label{sec:data}
{In this section, we summarize all publicly available datasets from past studies on modern slavery that utilized MSE (table \ref{table:datasets}).}
The \textbf{United Kingdom}, \textbf{New Orleans}, \textbf{Netherlands}, and \textbf{Western U.S.} datasets were considered in a recent study of \cite{Silverman2020}. All data considered has been stripped of covariate information, although in some of the original studies (Netherlands and Western U.S.) such covariate data was available to the researchers and used to produce estimates. Furthermore, we only consider a maximum of five lists for each of the datasets, in order to ensure applicability of the decomposable graphical model approach of \cite{Madigan1997, Lum2015}. For the United Kingdom, New Orleans, and Netherlands datasets which originally contained more than five lists, we consider the five lists version proposed in \cite{Silverman2020}. Full details on each dataset can be found in Appendix~\ref{sec:appendix_data}.

\begin{table}
\caption{\label{table:datasets}Datasets under consideration, the timeframe for the collected data, the number of observations, the number and proportion of overlap (observations which appeared in more than one list), and the total number of lists. In order from top to bottom, datasets are from \cite{Silverman2014}, \cite{Bales2019}, \cite{VanDijk2017}, \cite{Farrel2019}, and \cite{Lyneham2019}.}
 \centering
   \begin{tabular}{lcccc}
    \toprule
    Dataset                 & Timeframe     & \# observations   & \# overlap  & \# lists  \\\midrule
    \textbf{United Kingdom}          & 2013          & 2744              & 221 (8.1\%)           & 5        \\
    \textbf{New Orleans}            & 2016          & 185               & 12 (6.5\%)            & 5        \\
    \textbf{Netherlands}             & 2010--2015    & 8234              & 431 (5.2\%)           & 5         \\
    \textbf{Western U.S.}            & 2016          & 345               & 23 (6.7\%)           & 5         \\
    \textbf{Australia}               & 2015–16 to 2016–17 & 414          & 69 (16.7\%)           & 4         \\
    \bottomrule
\end{tabular}
\end{table}%

\subsection{Population Size Estimators}\label{sec:estimators}
{We now introduce} the population size estimators which we evaluate and compare. The choice is motivated by past MSE studies for the quantification of modern slavery, as well as by the recent comparative analysis of \cite{Silverman2020}. First, we consider the approach of 
\cite{Chan2020}, which we refer to as \textit{SparseMSE} following the name of the corresponding R package. This approach was motivated by the studies of \cite{Silverman2014, Bales2015, Bales2019}. It addresses the issues of non-overlapping lists and of model selection uncertainty. Second, we consider the approach of \cite{Madigan1997}, which we refer to as \textit{dga} following its implementation in the {dga} R package of \cite{Lum2015}. This approach was considered in \cite{Silverman2020} and provides a Bayesian model averaging approach to MSE. Third, we consider the approach of \citep{ManriqueVallier2016} which we refer to as \textit{LCMCR} following the name of the corresponding R package. This estimator was also considered in \cite{Silverman2020} and has been used in human rights statistics, leading to recent extensions \citep{Kang2020}. Finally, we consider a simple independence model as a baseline point of reference. Each approach provides point and interval estimates, relying on the modeling approaches reviewed in section \ref{sec:modeling-approaches} for model fitting and inference.

\begin{description}
\item[SparseMSE:] SparseMSE \citep{Chan2020} fits a log-linear model of the form \eqref{eq:log-linear-parameterization} with no three-way or higher interaction terms and with two-way interaction terms selected through forward stepwise $p$-value thresholding. A Poisson likelihood approximation to the multinomial data likelihood is used for mathematical convenience; this does not meaningfully change inferences \citep{Cormack1989}. Extended maximum likelihood is used for parameter estimation, accounting for non-overlapping lists.
The ``bias-corrected and accelerated'' bootstrap procedure of \cite{diciccio1996bootstrap} is used for the construction of confidence intervals while accounting for model selection uncertainty.
Crucially, SparseMSE relies on the assumption of no full-way interaction, meaning that $\gamma = 0$ in the log-linear parameterization \eqref{eq:log-linear-parameterization} of the model probabilities.

%The main tuning parameter of SparseMSE is the $p$-value threshold used for stepwise model selection. \cite{Bales2015} considered a threshold of $0.05$ and \cite{Chan2020} use a default threshold of $0.02$. A smaller threshold leads to bias towards the independence model, whereas a higher threshold allows for the consideration of more complex models with correspondingly higher estimator variance. Sensitivity of estimates to this choice of threshold is explored in section {[XX]}.

%Crucially, SparseMSE relies on the assumption of no full-way interaction, meaning that $\gamma = 0$ in the log-linear parameterization \eqref{eq:log_linear_parameterization} of the model probabilities. {[The consequences of this assumption are considered in ...]}

\item[dga:] The dga approach \citep{Madigan1997, Lum2015} uses decomposable graphical models with hyper-Dirichlet priors and Bayesian model averaging, as described in section \ref{sec:decomposable-graphical-models}, to obtain a population size posterior distribution. By default, ``prior counts'' are set to be constant and with value $2^{-L}$, where $L$ is the number of lists, the prior on the set of decomposable graphs is constant, and the population size prior is the improper prior $p(N) \propto 1/N$. 
%Sensitivity of estimates to these choices is explored in section {[XX]}.
%
As a population size estimator, we consider the median of the posterior distribution. Confidence intervals are obtained by taking equally tailed quantiles of the posterior distribution. 
This approach also makes the assumption of no full-way interaction, meaning that $\gamma = 0$ in the log-linear parameterization \eqref{eq:log-linear-parameterization} of the model probabilities.
%The decomposable graphical modeling approach of \cite{Madigan1997}, as considered here, is a particular case of log-linear modeling with no full-way interaction. It relies on assumptions \textbf{A1} and \textbf{A2}, as well as assumption \textbf{A3.1} introduced in \eqref{eq:assumptionA3.1}. As long as the dependencies between lists can be represented by a conditional independency graph which is not complete (where at least one edge is missing), then assumption \textbf{A3.1} is satisfied.

\item[LCMCR:] LCMCR \citep{ManriqueVallier2016} uses the latent class representation \eqref{eq:latent_class_model} together with a stick-breaking prior on the weights $w_k$ and a uniform prior on the list inclusion probabilities $\lambda_{k,j}$. The stick-breaking prior defines $w_1 \sim \text{Beta}(1,\alpha)$, $w_2 = (1-w_1)v_2$ with $v_2 \sim \text{Beta}(1, \alpha)$, $w_3 = (1-w_2)v_3$ with $v_3 \sim \text{Beta}(1, \alpha)$, and so forth, with $\alpha$ itself being Gamma distributed. By default,  $\alpha \sim \text{Gamma}(0.25, 0.25)$. The population size $N$ is given the default improper prior $p(N) \propto 1/N$. The posterior distribution of $N$ is approximated through conjugate Gibbs sampling. By default, we run 200 randomly initialized chains, each with $100,000$ iterations which are thinned down to $100$ samples. The number of latent classes is limited to a maximum number of $10$ classes to reduce computational burden. We summarize the population size posterior distribution using the posterior median and equally tailed quantiles for confidence intervals.
%
%$2,000,000$ iterations of the Gibbs sampler, discarding half as burn-in and thinning down to $10,000$ samples to approximate the posterior distribution. The number of latent classes is limited to a maximum number of $10$ classes. We summarize the population size posterior distribution using the posterior median and equally-tailed quantiles for confidence intervals.

%It is important to note that the LCMCR estimator makes no assumption of the form \textbf{A3}. The population size posterior distribution is therefore generally inconsistent --- the posterior distribution does not converge to the true population size. Instead, \texttt{LCMCR} quantifies uncertainty through the prior relationship between $p_{\bm{0}}$ and the other model probabilities. This is valid as long as the mixture of independence model is appropriate for the data, and as long as the prior specification properly captures our uncertainty regarding the distribution of the latent classes and the list inclusion probabilities.
\item[Independence:] Additionally, we consider an independence model as a baseline point of reference. This is a Poisson log-linear model with no two-way or higher interaction terms. It is fitted to the data through maximum likelihood, using the \texttt{modelfit()} function of the \texttt{SparseMSE} R package \citep{Chan2020}.
\end{description}

\begin{remark}
    \ob{There are many other models and estimators used for capture-recapture and multiple systems estimation \citep{Otis1978, Amstrup2005, baillargeon2007rcapture, laake2013marked, overstall2014conting, bohning2017capture, Worthington2021}. Our paper focuses on approaches which have previously been used in multiple systems estimation studies for the quantification of modern slavery and which are suited to modern slavery data. In comparison, many capture-recapture models from the population ecology literature require some amount of experimental control to justify strong underlying assumptions such as assumptions of independence between lists. {To our knowledge, experimental control is not present in the modern slavery applications that we consider, which could make the modeling assumptions ill-suited. Thus, we focus on previously-used techniques with more realistic assumptions.}} 
%     and modeling assumptions are difficult to justify, which is why we focus on previously-used techniques in this application.}
\end{remark}

\subsection{Comparison of Estimates}\label{sec:comparison_estimates}

{Figure \ref{fig:real-data-comparison} shows the comparison of the SparseMSE, dga, LCMCR, and Independence estimates on the datasets introduced in Section \ref{sec:data}.}

%In this section, we compare each of the estimates from the SparseMSE, dga, LCMCR, and Independence estimators.
%Figure \ref{fig:real-data-comparison} showcases the estimates applied to the modern slavery datasets introduced in Section \ref{sec:data}. 

All of the approaches considered, except for the independence model, account for model selection uncertainty, whereas published estimates from past studies did not. This explains the very narrow uncertainty in some recent published estimates \ob{when} compared to the estimates of the SparseMSE, dga, and LCMCR approaches. 

Observe that there is general agreement between the estimates. This is particularly pronounced in the case of the New Orleans and Western U.S. datasets for which little overlap between lists is available in the data. Only 12 cases appeared in more than one list in New Orleans, and only 23 cases appeared on more than one list the Western U.S. Without overlap in the data to accurately estimate interaction terms, the regularization implicit to these approaches tends to produce estimates that are in alignment with the independence model estimates. 

%\textcolor{red}{OB: Are you trying to state that all of the estimates are reasonable given the type of data we have here and perhaps one would prefer the ones with uncertainty quantification so that narrow estimates could be avoided in practical use. Then you could state that to explore each model further, you will explore the sensitivity of the methods/models. Give the reader an idea of what you have found more concretely and where you are heading.}

{Overall, the dga and LCMCR estimates tend to be comparable. This is not something which should be expected givent that the dga and LCMCR models rely on different identifying assumptions. However, both model contain the independence model as a particular case. The independence model estimates are similar to most other estimates, although it provides very narrow confidence intervals on larger datasets. The SparseMSE estimates notably differ in the cases of the United Kingdom and of Australia.}

{In section \ref{sec:sensitivity_convergence}, we explore some of the sensitivities of these approaches to tuning parameters which may influence the results. The sensitivity analysis highlights challenges associated with the use of each estimator.}

\begin{figure}[ht!]
    \centering
    \includegraphics{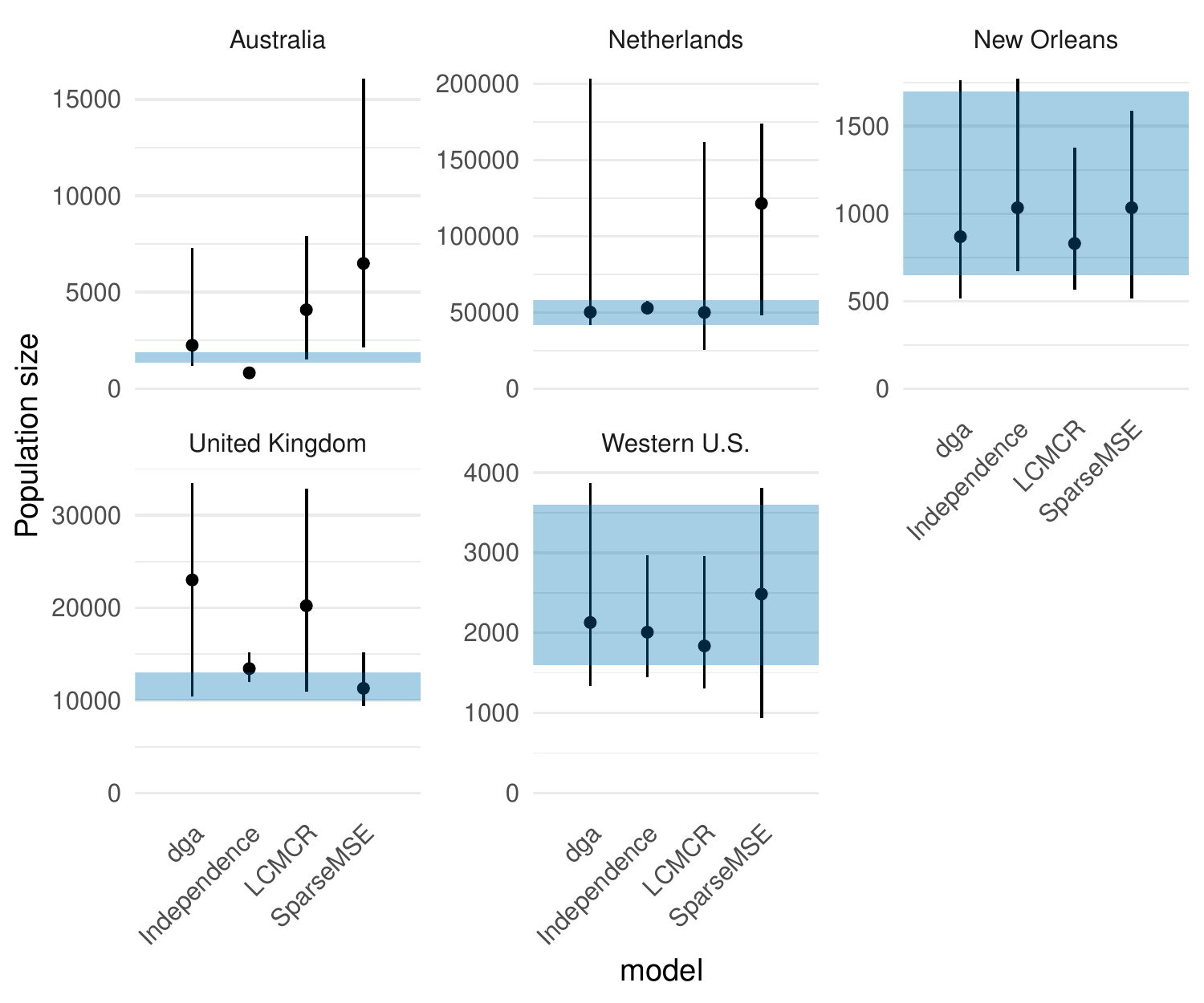}
    \caption{Comparison of the independence model, SparseMSE, LCMCR and dga estimates on the modern slavery datasets described in Section \ref{sec:data}. Vertical line ranges represent 95\% confidence intervals. The blue shaded regions represent the published estimates.}
    \label{fig:real-data-comparison}
\end{figure}

\subsection{Sensitivity and Convergence Issues}\label{sec:sensitivity_convergence}
%\textcolor{red}{OB: Summarize this section please.} \textcolor{red}{I would recommend starting with the sensitivity and ending with convergence issues. What about the independence model?}
{Let us now turn to the sensitivities of the SparseMSE and dga estimates to choices of tuning parameters, as well as convergence issues with LCMCR. The sensitivities and convergence issues can be major challenges in practice, even before potential issues with the data and models are considered. Note that we focus on convergence issue with LCMCR, rather than its sensitivity to choices of priors, as this is, we believe, the most important problem which the approach currently faces. We also ignore the Independence estimator which requires no tuning parameters.

{Throughout, we focus on highlighting issues on datasets for which they are most noticeable or most relevant. Our goal is to showcase issues which \textit{can} happen and which should be evaluated in practice, rather than to provide an analysis for each of the five considered datasets.}}

%\subsubsection{Sensitivity of Independence estimates}

%\textcolor{red}{Are these not sensitive at all? Is there nothing to say regarding these models? It seems that not including some sensitivity analysis is going to be misleading. If they're not sensitive, then wouldn't we use them all the time! You sort of get into this below, but it's not very clear.}

\subsubsection{Sensitivity of SparseMSE Estimates}
{We consider} the main tuning parameter of SparseMSE --- the $p$-value threshold used for stepwise model selection. \cite{Bales2015} used a threshold of $0.05$ and \cite{Chan2020} used a threshold of $0.02$. Intuitively, we would expect smaller threshold to lead to bias towards the independence model, whereas a higher threshold allows for the consideration of more complex models with correspondingly higher estimator variance. This is not necessarily the case. Figure \ref{fig:sparsemse_sensitivity} showcases the SparseMSE estimates on the United Kingdom dataset for $p$-value thresholds between $0$ and $0.1$. While the threshold of $0.02$, used in \cite{Chan2020}, led to the narrow confidence interval of between $10,000$ and $15,000$ potential victims, a very slightly \textit{smaller} threshold leads to the much bigger interval of between $10,000$ and $30,000$ potential victims. This is surprising behavior - smaller thresholds should correspond to higher regularization, but they can unexpectedly produce much larger confidence intervals \ob{due to higher model selection uncertainty.}

\begin{figure}[!ht]
    \centering
    \includegraphics{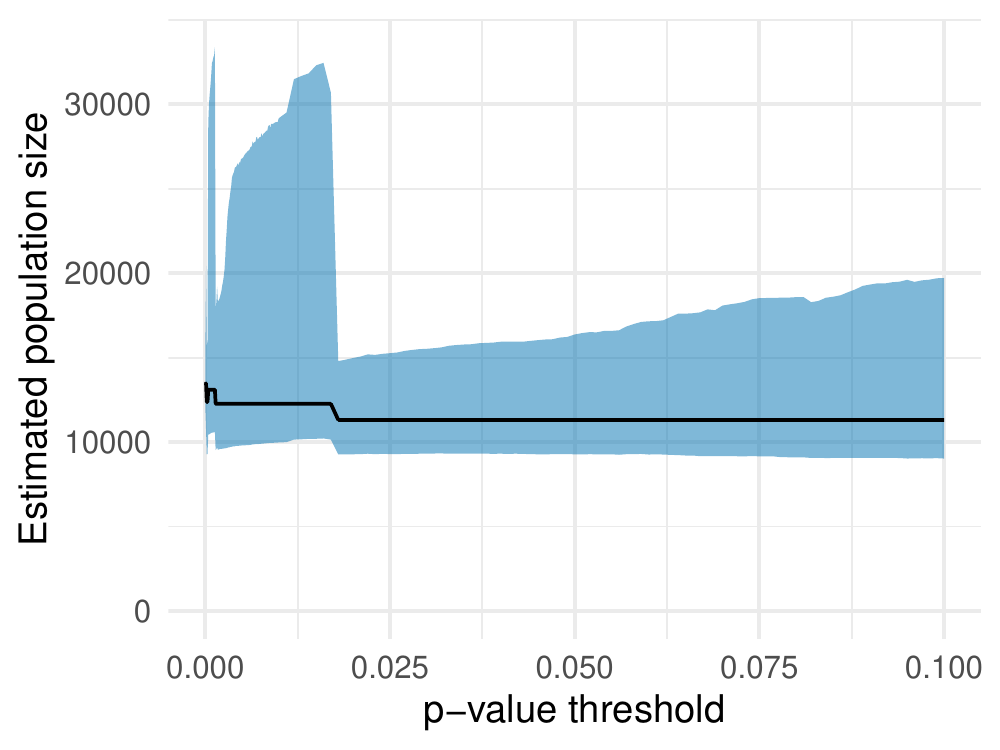}
    \caption{Estimates of the SparseMSE approach on the \textbf{United Kingdom} data, for $p$-value thresholds between $0$ and $0.1$. The black line represents the point estimates and the blue band represents the $95\%$ bootstrap confidence intervals.}
    \label{fig:sparsemse_sensitivity}
\end{figure}

In practice, \ob{larger $p$-value thresholds might be preferable given that they allow the estimation of more complex interactions between lists.} \ob{However, if there is no strong justification for the use of one $p$-value threshold over another,} we recommend that the range of estimates corresponding to different thresholds be investigated and reported in applications.

%there is no strong justification for the use of one $p$-value threshold over another. Given the unexpected sensitivity of estimates to this parameter, we recommend that the range of estimates corresponding to different thresholds be investigated and reported in applications. }

\subsubsection{Sensitivity of dga Estimates}

{Let us now turn to dga estimates. We explore their sensitivity to reasonable choices of hyperparameters or prior distributions.}
There are two prior distributions for which it may be difficult to elicit informative priors and which we focus on. Specifically, we consider the choice of ``prior counts'' which determine the hyper-Dirichlet prior on decomposable graphs and the choice of prior on the set of decomposable graph structures. 

First, regarding prior counts, \cite{Lum2015} proposes the default $\delta = 2^{-L}$ where $L$ is the number of lists. The choice of $\delta = 2^{-L}$ corresponds to the expected count under an independence model for which each list has an inclusion probability of $1/2$. We extend this prior by considering the expected count under an independence model where each list has an inclusion probability $\kappa > 0$. Values $\kappa < 0.5$ corresponds to lower probabilities of inclusion on individual lists, which seems more reasonable in the context of modern slavery data. 

Second, regarding the prior on the graphical structure, we consider a ``small-world'' prior restricted to decomposable graphs. That is, each edge in the graph appears with independent probability $\beta \in (0,1)$, conditionally on the resulting graph being decomposable (and with the complete graph being excluded as well). Larger values of $\beta$ gives more weight to more complex models, while smaller give more weight to less complex models. The default uniform prior corresponds to setting $\beta = 1/2$.

\begin{figure}[!ht]
    \centering
    \includegraphics{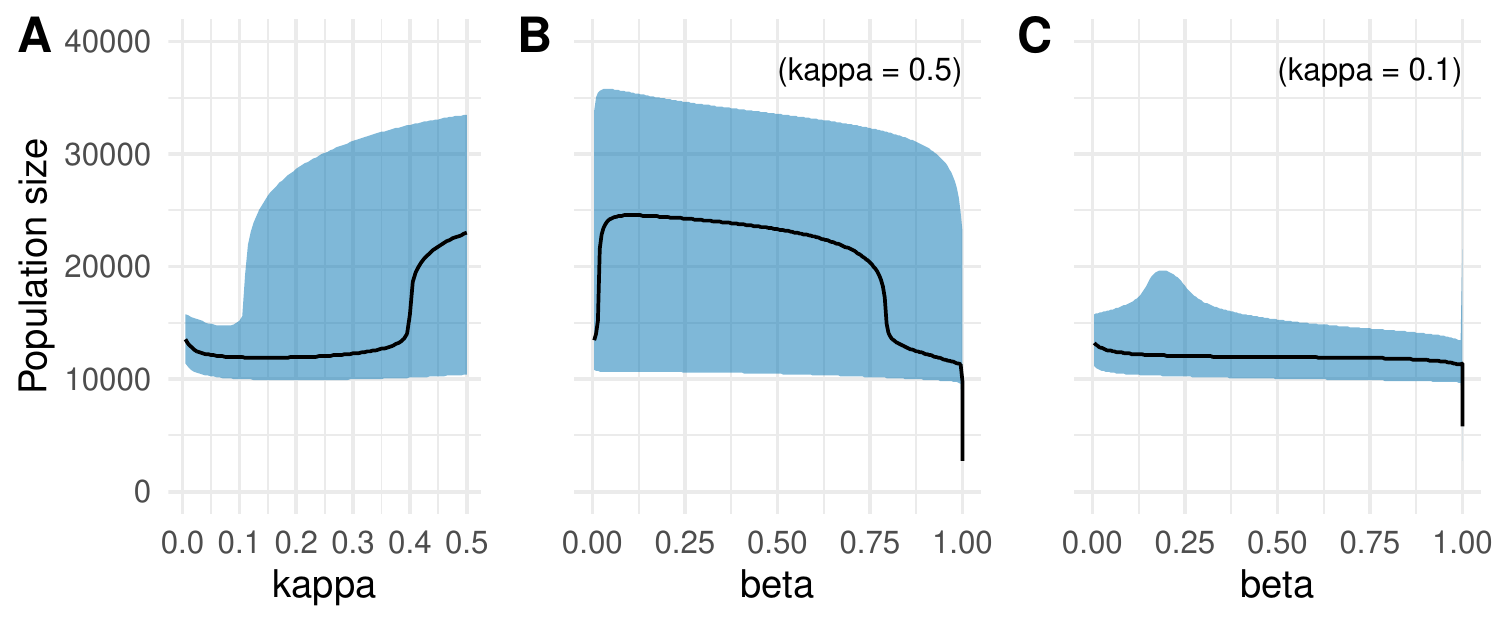}
    \caption{Point estimates and $95\%$ credible intervals {of dga estimates for the \textbf{United Kingdom} dataset,} with varying values of $\kappa$ (which parameterizes the prior counts) and $\beta$ (which parameterizes the distribution on graphical structures).  In panel \textbf{A}, $\beta$ is fixed to a value of $1/2$ and $\kappa$ ranges between $0$ and $0.5$. In panel \textbf{B}, $\kappa = 0.5$ and $\beta$ ranges between $0$ and $1$. In panel \textbf{C}, $\kappa = 0.1$ and $\beta$ ranges between $0$ and $1$. 
    %{Here the range of plausible estimates is comparable to the estimates obtained with default parameters, as in Figure \ref{fig:real-data-comparison}. The priors can still be highly influencial, especially when considering the prior counts parameterized by $\kappa$.}
    }
    \label{fig:dga_sensitivity}
\end{figure}

%\textcolor{red}{Need to state more regarding how sensitive these estimates are to the reader. They're very sensitive to the choice of the tuning parameters.}
Figure \ref{fig:dga_sensitivity} shows estimates resulting from different choices of $\kappa$ and $\beta$ {in application to the United Kingdom dataset. We see that prior choices can quite heavily influence estimates and the width of confidence intervals. In practice, unless a particular prior can be rigorously justified, estimates should be reported for the whole range of plausible prior distributions. This can be viewed as part of an ``objective Bayesian analysis,'' where we acknowledge the difficulty of selecting a prior distribution.} 

\subsubsection{LCMCR Convergence Issues}

%\textcolor{red}{Does LCMCR have any sensitivity issues to tuning parameters? If not, this is worth mentioning as this is a point in its favor. Then one could start out by saying while LCMCR is robust, it unfortunately may have other issues. } All Bayesian approaches are sensitive to tuning parameters - basically you can obtain any result with the "right" choice of prior. But the main issue with LCMCR is the convergence issue. It's kind of pointless to go into sensitivity to prior choice when MCMC doesn't converge, IMHO.

{Like any other Bayesian approach, LCMCR estimates are sensitive prior choices and ranges of reasonable priors should be explored in practice. However, LCMCR faces an important additional challenge -- the MCMC algorithm used to compute estimates does not converge in some cases. This problem is due to the non-identifiability of the latent class model which results in a multimodal posterior distribution. With large datasets, {such as the Netherlands data}, the Gibbs sampling algorithm of \cite{ManriqueVallier2016} can struggle to explore the posterior distribution.}

Figure \ref{fig:lcmcr_trace} provides the Markov chain Monte Carlo (MCMC) samples used to approximate the LCMCR posterior distribution of the non-observation probability in application to the Netherlands dataset. Observe one trace plot for 200 independent chains, where each chain was run for $100,000$ iterations and thinned down to 100 samples. We find that there are two posterior modes and a lack of mixing between them. 
%One trace is plotted for each of the 200 independent chain, where each chain was run for $100,000$ iterations and thinned down to 100 samples.
We provide MCMC convergence diagnostics in table \ref{tab:diagnostics} (left) for the non-observation probability $p_0$, for the number of unobserved individuals $n_{\text{obs}}$, and for the number of latent classes $k^*$. The $\hat R$ value \citep{carpenter2017stan, gelman2013bayesian} of $1.67$ for the non-observation probability, as well as the effective sample size $n_\text{eff}$ of only $340$ for the total number or $20,000$ samples across chains, is witness to non-convergence. Using $20$ chains, each running $1000$ times longer and thinned down to $1,000$ samples, results in a lower effective sample size. Anecdotically, we have not been able to run the Gibbs sampler long enough to observe proper mixing of the non-observation probability.

Given this lack of convergence, we can use a large number of randomly initialized parallel chains to ensure stability of estimates across replications. This explains our default choice of $200$ independent chains in our analyses. Other more sophisticated approaches can be used to deal with peaked and multimodal posteriors which mix poorly, such as parallel tempering \citep{earl2005parallel}, using parallel chains \citep{gelman1992inference}, and more \citep{yao2020stacking}. Implementing these approaches would be necessary for the application of LCMCR to larger datasets. We only encounter this issue for the Netherlands dataset, and given the scope of our paper, we leave this for future work.

%In our context we have only encountered issues with the Netherlands data, and so we do not explore this further.

\begin{figure}[!ht]
    \centering
    \includegraphics{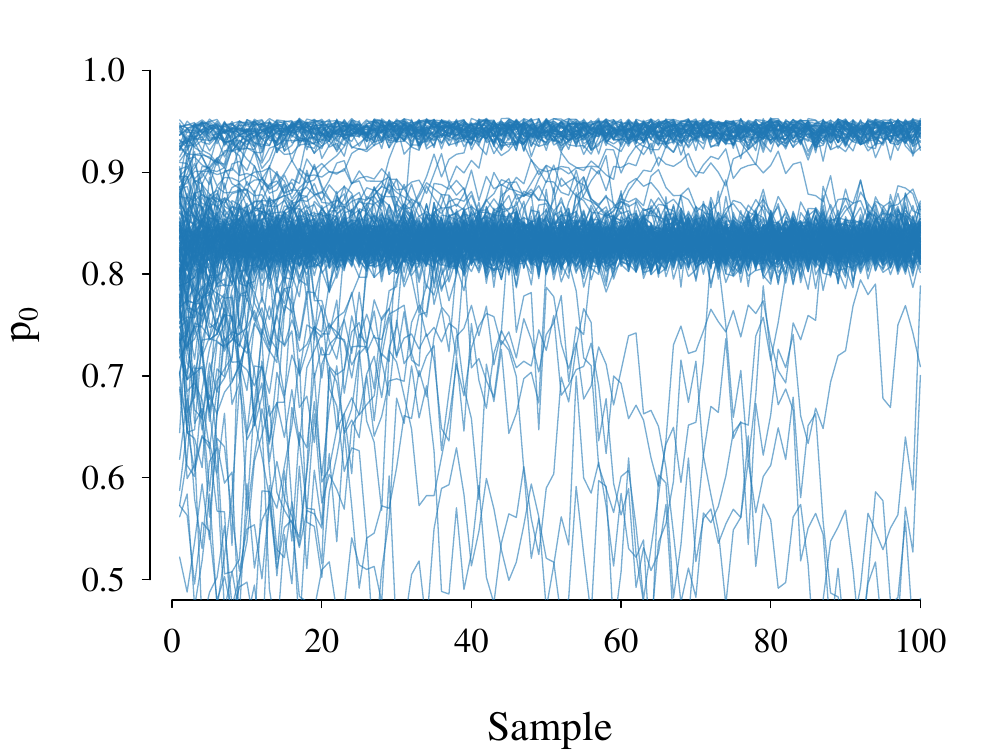}
    \caption{MCMC traces of the non-observation probability $p_0$ for $200$ independent chains using the Gibbs sampler of \cite{ManriqueVallier2016}, applied to the \textbf{Netherlands} dataset. 
    %Notice the two modes of the posterior distribution and the lack of mixing between them.
    }
    \label{fig:lcmcr_trace}
\end{figure}

\begin{table}
\caption{\label{tab:diagnostics}Convergence diagnostics for LCMCR samples aggregated across chains, both for our default settings (left) and for $20$ chains each run $1000$ times longer than by default (right). Here $n_0$ represents the number of unobserved individuals, $p_0$ is the non-observation probability, and $k^*$ is the number of latent classes. The large $\hat R$ values and low effective sample sizes are indicative of poor MCMC mixing.}
\centering
\begin{tabular}{rrrr}
\multicolumn{4}{c}{\textbf{200 chains of $10^5$ iterations}}\\
  \toprule
 & $\hat R$ & $n_\text{eff}$ & $n_\text{samples}$ \\ 
  \midrule
$n_0$ & 1.67 & 340.84 & $2\cdot 10^4$ \\ 
  $p_0$ & 1.67 & 340.48 & $2\cdot 10^4$ \\ 
  $k^*$ & 1.39 & 455.64 & $2\cdot 10^4$ \\ 
   \bottomrule
\end{tabular}
\hspace{0.5in}
\begin{tabular}{rrrr}
\multicolumn{4}{c}{\textbf{20 chains of $10^8$ iterations}}\\
  \toprule
 & $\hat R$ & $n_\text{eff}$ & $n_\text{samples}$ \\ 
  \midrule
  $n_0$ & 1.46 & 38.04 & $2\cdot 10^5$ \\ 
  $p_0$ & 1.46 & 38.03 & $2\cdot 10^5$ \\ 
  $k^*$ & 1.01 & 2338.91 & $2\cdot 10^5$ \\ 
   \bottomrule
\end{tabular}
\end{table}

\section{Internal Consistency Analysis}\label{sec:internal_consistency_analysis}

{We now turn to our first analysis of the accuracy of MSE estimates in application to data from modern slavery studies. This analysis relies on subsets of the data for which ``ground'' truth is available, {which means that the true population size is already known.} This was termed an ``internal consistency analysis'' by \cite{Hook2012} (see also \cite{hook2000accuracy, Brittain2009}). This provides a way to evaluate the accuracy of MSE on relevant datasets. The way in which ground truth data is obtained is described in section \ref{sec:ground_truth_data} and the performance of MSE estimators on this data is described in section \ref{sec:results_internal_consistency}. Limitations of this approach which motivate the rest of our paper are discussed in section \ref{sec:internal_consistency_limitations}.}

%In this section, we evaluate the accuracy of estimators on data where ground truth is available, following ``internal consistency analysis'' \citep{Hook2012}. First,  we describe how data with ground truth is obtained from the datasets considered in Section \ref{sec:data}. Next, we showcase the results of the analysis using performance metrics proposed by \cite{Hook2012}. Finally, we discuss limitations of this approach.
%regarding the generalizability of results.

\subsection{Ground Truth Data Through Conditioning}\label{sec:ground_truth_data}

{To illustrate how we obtain data with ground truth, consider the United Kingdom dataset reproduced in table \ref{table:UKdata}.} This dataset contains five lists, including the local authorities (LA) list. Conditioning on cases being recorded by the LA list and omitting the LA list itself, we obtain the conditioned data shown in table \ref{tab:UK_conditioned}. In addition, we know that a total of $94$ cases have appeared on the LA list. We may therefore attempt to use the conditioned data, which record $40$ cases having appeared on the LA list as well as on other lists, in order to estimate the total of $94$ cases which appeared on the LA list. Here the LA list is our \textit{reference} list, and $94$ is the ground truth population size for the conditioned data.

\begin{table}
\caption{\label{tab:UK_conditioned} United Kingdom dataset conditioned on the LA list. Note that no cases appeared on both the LA list and the GP lists.}
\centering
\begin{tabular}{r|ccc|ccc}
\toprule
  & 15 & 19 & 3 & 1 & 1 & 1\\
\midrule
NG & $\times$ &  &  & $\times$ & $\times$ & $\times$\\
PFNCA &  & $\times$ &  & $\times$ &  & $\times$\\
GO &  &  & $\times$ &  & $\times$ & $\times$\\
GP &  &  &  &  &  & \\
\bottomrule
\end{tabular}
\end{table}

\subsection{Analysis and Results}\label{sec:results_internal_consistency}

The process of using a reference list to obtain conditioned data and a corresponding ground truth is repeated \textbf{for every dataset in table \ref{table:datasets} and for every list.} Datasets with fewer than 30 observations are discarded, resulting in the total of $11$ conditioned datasets described in table \ref{tab:conditioned_data}. The SparseMSE, dga, LCMCR and Independence estimators are then applied to these datasets, and the point estimates $\hat N$ are compared to the ground truth population size $N$, resulting in the log relative bias, 
$
\log (\hat N /N).
$

In table \ref{tab:internal_consistency_result}, we report its empirical mean $\mathbb{E}[\log \hat N /N]$, its root mean square error (RMSE) $\mathbb{E}[(\log \hat N /N)^2)]$, and its median, after removing the outlying results of Netherlands' list K (for which no overlap data is available). Additionally, we report the empirical coverage of $95\%$ confidence intervals. Figure \ref{fig:internal_consistency_results} shows the estimates and ground truth for every conditioned dataset.

\begin{table}
\caption{\label{tab:conditioned_data}Description of conditioned datasets with more than 30 observations.}
\centering
\begin{tabular}{lcccc}
\toprule
Dataset & Reference list & Ground truth & \# observations & \# overlap\\
\midrule
 & LA & 94 & 40 & 3\\
 & NG & 567 & 104 & 7\\
 & PFNCA & 1169 & 174 & 6\\
\multirow{-4}{*}{\raggedright \textbf{United Kingdom}} & GO & 807 & 112 & 6\\
\cmidrule{1-5}
 & IO & 929 & 173 & 13\\
 & K & 1348 & 49 & 0\\
 & P & 4812 & 346 & 14\\
 & R & 742 & 92 & 3\\
\multirow{-5}{*}{\raggedright \textbf{Netherlands}} & Z & 848 & 216 & 12\\
\cmidrule{1-5}
 & B & 77 & 64 & 23\\
\multirow{-2}{*}{\raggedright \textbf{Australia}} & C & 260 & 62 & 22\\
\bottomrule
\end{tabular}
\end{table}

The SparseMSE point estimate appears to perform best, with low mean and median log relative bias. Otherwise, the point estimate are all roughly comparable. Regarding confidence intervals, lower bounds are smaller than the ground truth in all cases (excepted with SparseMSE applied to the Netherlands dataset conditioned on list $K$). It is interesting to note that the Independence model does not perform worse than other approaches.

\begin{table}
\caption{\label{tab:internal_consistency_result}Summary results of the internal consistency analysis. Best results are bolded in each column.}
\centering
\begin{tabular}{lcccc}
  \toprule
Estimator & Mean & RMSE & Median & Coverage \\ 
  \midrule
dga & -0.34& 0.60 & -0.22 & 0.80 \\ 
Independence & -0.29 & \textbf{0.55} & -0.28 & 0.88 \\ 
LCMCR & -0.52 & 0.72 & -0.50 & 0.60 \\ 
SparseMSE & \textbf{-0.17} & 0.63 & \textbf{-0.15} & \textbf{0.90} \\ 
   \bottomrule
\end{tabular}
\end{table}

\begin{figure}[h!]\label{fig:internal_consistency_results}
    \centering
    \includegraphics{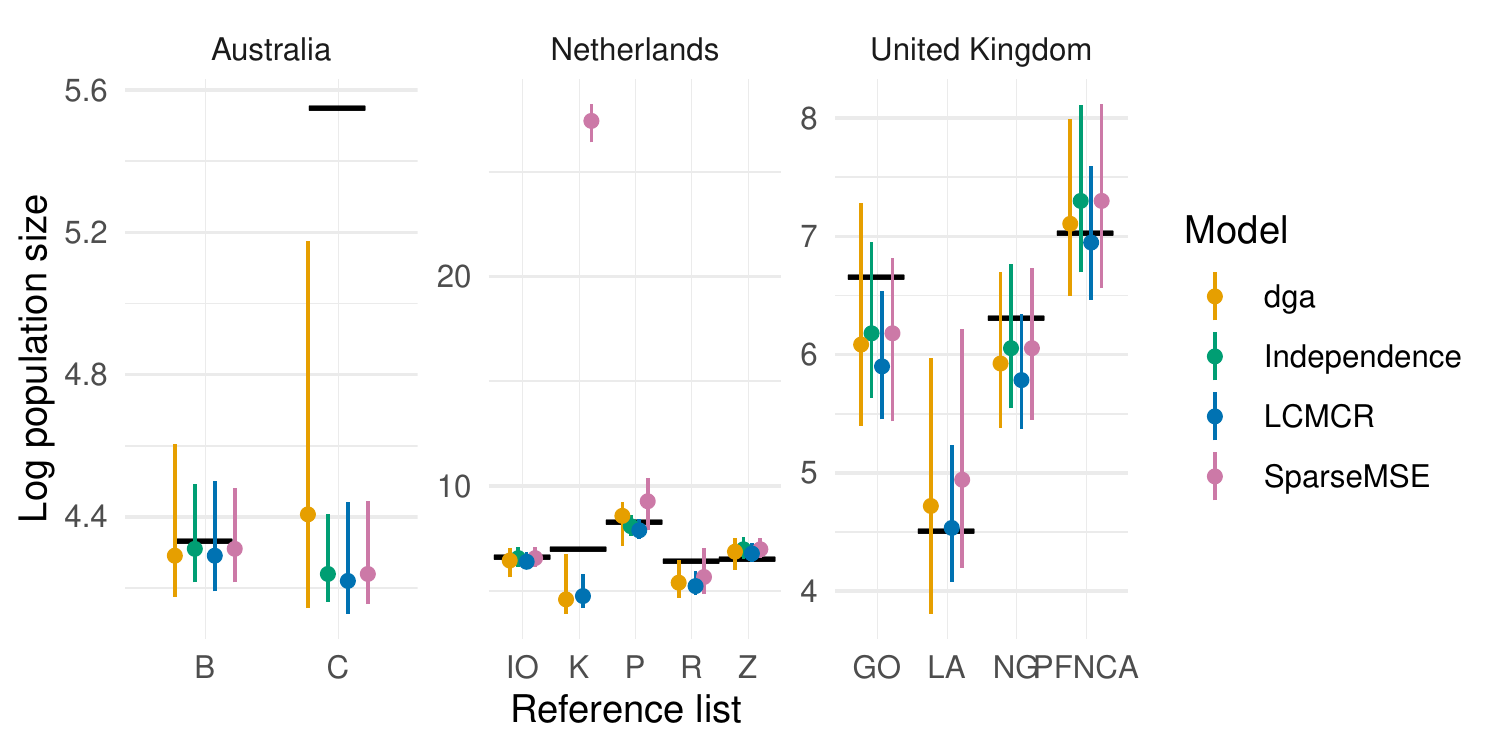}
    \caption{Results of the internal consistency analysis for every considered dataset and reference list. {The black horizontal lines represent ground truth population size. Estimates and $95\%$ confidence intervals are represented by points and vertical lines.}}
    \label{fig:internal_consistency}
\end{figure}

{Overall, the results of the internal consistency analysis are highly encouraging. Lower bounds of the confidence intervals are almost always lower than ground truth, and estimates tend to be close to the ground truth. Coverage of SparseMSE is almost nominal at $90\%$. }

\subsection{Limitations}\label{sec:internal_consistency_limitations}

%\textcolor{red}{In this section, we discuss limitations regarding using this approach, and proposed solutions.}
{The main limitations of the internal consistency analysis are} that the conditioned datasets for which ground truth is available are few in number and are not entirely representative of the modern slavery application. There are two main issues here: (1) the conditioned datasets are small and contain little overlap data; and (2) conditioning on a large list necessarily removes from consideration the features of unobserved individuals. That is, regarding point (2), issues of individual heterogeneity and of certain interaction between lists may be removed by conditioning. Our observations regarding the accuracy of estimates on conditioned datasets may therefore not be generalizable to real applications.

Regarding issue (2), section \ref{sec:large-sample-bias} next evaluates the bias which can be expected in the presence of individual heterogeneity. {That is, we provide a novel characterization of the bias of MSE estimators when underlying assumptions are not satisfied. This characterization is used to compute the bias of estimates under various heterogeneity models.} In order to address issue (1), we propose in section \ref{sec:visual_assessment} a visual resampling technique to  evaluate the robustness of estimates on practical datasets.

\section{Bias Under Misspecified Assumptions}\label{sec:large-sample-bias}

Throughout the paper, we have mentioned how population size estimation relies on an untestable assumption of the form \textbf{A3}. 
In this section, we investigate assumption \textbf{A3.1} {of no full-way interaction term in} the log-linear model.
We consider the set of estimators which are consistent under \textbf{A3.1}, used in modern slavery studies, and we characterize their asymptotic bias when this assumption is misspecified (Theorem \ref{thm:characterization}). In section \ref{sec:bias_heterogeneity}, we consider the consequences of individual heterogeneity on the bias of population size estimators. {We chose heterogeneity models for their relevance to the modern slavery application, as individual characteristics necessarily affect list inclusion probability.} Proposition \ref{prop:2} describes the sign of the bias under a general heterogeneity model. We discuss the case of a Beta heterogeneity model in section \ref{sec:beta-heterogeneity}. Figure \ref{fig:beta_model_multi_list} illustrates the magnitude of the bias under the Beta heterogeneity model as a function of a precision parameter and of the number of lists. The results are summarized in section \ref{sec:theory_summary}.

\subsection{Characterization \ob{of the} Asymptotic Relative Bias}

\ob{A standard} asymptotic framework for MSE \citep{Chao2008} considers the large population limit $N \rightarrow \infty$, where the list inclusion patterns $\{W_i\}_{i=1}^\infty$ (see Section \ref{sec:general_framework}) are independent and have distribution \eqref{eq:def_px}. In addition, the counts $\{n_x\}_{x \not = \bm{0}}$ are now a function of $N$ as specified by \eqref{eq:def_counts}. In this context, we can define the consistency property of population size estimators $\hat N,$ which are functions of $\{n_x\}_{x \not = \bm{0}}$.

\begin{definition}[Consistency]
    A population size estimator $\hat N$ is said to be \textit{consistent}, for the model specified by \eqref{eq:def_px}, if $\hat N /N \rightarrow 1$ almost surely as $N \rightarrow \infty$.
\end{definition}

Departure from consistency is quantified by the asymptotic relative bias.

\begin{definition}[Asymptotic relative bias]
    The asymptotic relative bias of a population size estimator, for the model specified by \eqref{eq:def_px}, is defined as
    \begin{equation}
        \lim_{N \rightarrow \infty} \frac{\hat N - N}{N}
    \end{equation}
    when this limit is well-defined and almost surely constant.
\end{definition}

As noted in Proposition \ref{prop:1}, no population size estimator is consistent for all models. Theorem \ref{thm:characterization} characterizes the asymptotic bias of all estimators which would be consistent under assumption \textbf{A3.1}, when in fact this assumption is misspecified. {Note that all convergence statements are understood to happen with probability one (i.e. almost surely).}

%The proof of the Theorem is in the Appendix.

%Note that no population size estimator is consistent for all choices of probabilities $p_x \geq 0$ in \eqref{eq:def_px}. An assumption, such as the assumption $\gamma = 0$ of no full-way interaction, must be made for the population size $N$ to be estimable. Theorem \ref{thm:characterization} below characterizes the asymptotic relative bias of all estimators which would be consistent if it were indeed the case that $\gamma = 0$. The proof of the Theorem is in the Appendix.

\begin{theorem}[Characterization of the asymptotic relative bias]\label{thm:characterization}
    Let $\hat N$ be any population size estimator which is consistent under assumption \textbf{A3.1} of $\gamma = 0$ (no full-way interaction) in the log-linear representation of the model. Then the relative asymptotic bias of $\hat N$ exists and is given, when $\gamma$ is not necessarily equal to zero, by
    \begin{equation}\label{eq:thm_characterization}
        \lim_{N \rightarrow \infty} \frac{\hat N - N}{N}=p_{\bm{0}}(e^\gamma - 1).
    \end{equation}
\end{theorem}

{See Appendix~\ref{sec:proof_thm1} for the proof of Theorem \ref{thm:characterization}.}

\begin{remark}[Lower and upper bound estimators]
Theorem \ref{thm:characterization} shows that population size estimators which would be consistent when $\gamma = 0$ become \textit{lower bound estimators} when $\gamma < 0$ and \textit{upper bound estimators} when $\gamma > 0$. That is, when $\gamma < 0$ or $\gamma > 0$, $\hat N$ becomes a consistent estimator of a portion or of a multiple of $N$.
\end{remark}

Theorem \ref{thm:characterization} can be equivalently expressed as providing a first-order \ob{approximation to population size estimators.}

\begin{corollary}
    Any population size estimator $\hat N$ which is consistent in the absence of full-way interaction term between the lists has the approximation
    \begin{equation}
        \hat N = \left( 1 + \tfrac{p_{\bm{0}}}{1-p_{\bm{0}}}e^\gamma + o(1)\right)n_{\text{obs}}
    \end{equation}
    where $\gamma$ is defined in \eqref{eq:def_gamma_sum} and $o(1)$ is a term which tends to zero as $N \rightarrow \infty$.
\end{corollary}

\begin{proof}
    {Using the fact that $\lim_{N \rightarrow \infty} n_{\text{obs}}/N = 1 - p_{\bm{0}}$, we can rearrange \eqref{eq:thm_characterization} as
    $$
        \lim_{N \rightarrow \infty} \frac{\hat N}{n_{\text{obs}}} = 1 + \tfrac{p_{\bm{0}}}{1-p_{\bm{0}}}e^\gamma.
    $$
    Equivalently, $\frac{\hat N}{n_{\text{obs}}} = 1 + \tfrac{p_{\bm{0}}}{1-p_{\bm{0}}}e^\gamma + o(1)$ and the result follows directly.}
\end{proof}

\subsection{Bias in the Presence of Individual Heterogeneity}\label{sec:bias_heterogeneity}

%Theorem \ref{thm:characterization} can be used to characterize the asymptotic bias of population size estimators (assuming no highest order interaction term) in a variety of scenarios. 
{We now use Theorem \ref{thm:characterization} to consider the consequences of individual heterogeneity on the bias of estimators which assume no full-way interaction among lists. {Recall that these estimators are the ones being used in the context of modern slavery studies. Since individual heterogeneity is to be expected in this application, it is important to evaluate its practical consequences. In this section, we show that individual heterogeneity is incompatible with the assumption of no full-way interaction among lists. Also, we precisely quantify the effect of reasonable heterogeneity models on the bias of estimates.}}

%
%Previous work on the consequences of individual heterogeneity has been limited to particular cases. For instance, \cite{Hwang2005} show that Horvitz-Thompson type estimators have a downward bias in the presence of individual heterogeneity. In proposition \ref{prop:2}, we show that estimators based on \textbf{A3.1} can have either downward or upward bias under individual heterogeneity. In section \ref{sec:beta-heterogeneity}, this bias is precisely quantified for the particular case of a Beta heterogeneity model.
%
%The general heterogeneity model we consider is the following. 
Assume that each individual $i = 1, 2, \dots, N$ has an individual list appearance probability
\begin{equation}\label{eq:beta_heterogeneity}
    \lambda_i \sim^{ind.} F
\end{equation}
where $F$ is a distribution supported on $(0,1]$, and
\begin{equation}
    \mathbb{P}(W_i = x \mid \lambda_i) = \lambda_i^{\lvert x \rvert} (1-\lambda_i)^{L - \lvert x \rvert}, \quad x = (x_1, \dots, x_L) \in \{0,1\}^L
\end{equation}
where $\lvert x \rvert = \sum_{i=1}^L x_i$.
The inclusion patterns $W_i$ are still independent and they are marginally distributed as
\begin{equation}\label{eq:heterogeneity_model}
    p_x = \mathbb{P}(W_i = x) = \mathbb{E}\left[\lambda_i^{\lvert x \rvert}(1-\lambda_i)^{L-\lvert x \rvert}\right].
\end{equation}
{Using the notations of \cite{Otis1978}, this is termed an $M_{h}$ model allowing individual-specific inclusion probabilities.}

Proposition \ref{prop:2} illustrates some of the consequences of ignoring heterogeneity \ob{for estimators which are consistent under the assumption of no full-way interaction term.}
In the context of two lists, heterogeneity implies a negative bias. With more than two lists, the bias may be positive or negative. \ob{This can be constrasted with the behavior of other classes of estimators. For instance, Horvitz-Thompson type estimators which wrongly ignore heterogeneity \textbf{always} have a negative bias for any number of lists \citep{Hwang2005}.}

\begin{proposition}\label{prop:2}
    Let $\hat N$ be a population size estimator which is consistent under the assumption $\gamma = 0$ in \eqref{eq:def_gamma_sum}. In the context of two lists ($L = 2$) and for the latent heterogeneity model \eqref{eq:heterogeneity_model}, necessarily $\lim_{N \rightarrow \infty} \frac{\hat N - N}{N} \leq 0$. With three or more lists, the asymptotic bias is positive in some cases and negative in others.
\end{proposition}

\begin{proof}
    From Theorem \ref{thm:characterization}, it suffices to compute $\gamma$ in the context of the heterogeneity model \eqref{eq:heterogeneity_model}. In the two lists setting,
    $$
        \gamma = \log \frac{\left(\mathbb{E}\left[\lambda_i (1-\lambda_i)\right]\right)^2}{\mathbb{E}\left[\lambda_i^2\right] \mathbb{E}\left[ (1-\lambda_i)^2\right]} < 0
    $$
    by the Cauchy-Schwartz inequality, and hence $\hat N$ is negatively biased, asymptotically. With three and four lists, it is easy to find examples where the bias is positive or negative.
\end{proof}

\subsubsection{Beta Heterogeneity Model}\label{sec:beta-heterogeneity} 

In order to make Proposition \ref{prop:2} more concrete, consider the case where
$$
    \lambda_i \sim^{i.i.d.} \text{Beta}(a, b),
$$
and again
\begin{equation}
    \mathbb{P}(W_i = x \mid \lambda_i) = \prod_{i=1}^L \lambda_i^{x_i}(1-\lambda_i)^{1-x_i}, \quad x = (x_1, \dots, x_L) \in \{0,1\}^L.
\end{equation}
The inclusion patterns $W_i$ are marginally distributed as
\begin{equation}
    p_x = \mathbb{P}(W_i = x) \propto \Gamma(a + \lvert x \rvert) \Gamma(b + L - \lvert x \rvert).
\end{equation}
Furthermore, 
\begin{equation}\label{eq:beta_approximation}
    p_{\bm{0}} = \frac{\Gamma(a+b)\Gamma(b+L)}{\Gamma(b) \Gamma(a+b+L)} \approx \left(\frac{b}{a+b}\right)^L
\end{equation}
and
$$
    \gamma = -\sum_{k=0}^L (-1)^k {L \choose k} \log\left(\Gamma(a +k) \Gamma(b + L - k)\right).
$$
%Applying Theorem \ref{thm:characterization}, the asymptotic relative bias of any population size estimator which wrongly assumes $\gamma = 0$ is then
%\begin{equation}
%    \lim_{N \rightarrow \infty} \frac{\hat N - N}{N} = p_{\bm{0}}(e^\gamma - 1).
%\end{equation}

\paragraph{Two-Lists Beta Model}
In the context of two lists, where $L=2$, Theorem \ref{thm:characterization} simplifies as
$$
  \lim_{N \rightarrow \infty} \frac{\hat N - N}{N} = - p_{\bm 0}\left( \frac{a+b+1}{(a+1)(b+1)}\right) \leq - p_{\bm 0}\left( \max\left\{\frac{1}{a+1}, \frac{1}{b+1}\right\} \right) \leq 0.
$$
For example, with $a=1$ and $b=8$, it follows that $20\%$ of the cases are observed on average. The asymptotic relative bias is $-4/9$ and $\hat N \approx \frac{5}{9}N$. As $a \rightarrow 0$, the asymptotic relative bias tends towards $-100\%$. 

\paragraph{Three-Lists Beta Model}

In the context of three lists, we obtain
$$
    \gamma = \log \left(\frac{a(b+1)^2(a+2)}{b(a+1)^2(b+2)}\right).
$$
This is positive when $\mathbb{E}[\lambda_i] > 1/2$ and negative when $\mathbb{E}[\lambda_i] < 1/2$; there is a positive bias in the first case and a negative bias in the second. Note, however, that this simple expression for the sign of the bias does not hold outside of the Beta model. In general, the sign of the bias is also linked to higher moments of $\lambda_i$.

\paragraph{Multi-Lists Beta Model}

Now consider a Beta model with $L$ lists, where using \eqref{eq:beta_approximation} we fix $p_{\bm{0}} \approx (b/(a+b))^L = 3/4$ and we let the precision parameter $a+b$ of the list inclusion probabilities $\lambda_i$ vary. As $a+b$ tends to infinity, individual heterogeneity is reduced, while small positive values of $a+b$ represent high heterogeneity.

Figure \ref{fig:beta_model_multi_list} shows the asymptotic relative bias of population size estimators as a function of the precision $a+b$ of the Beta distribution and of the number of lists $L$. Note that there is a significant reduction of the relative bias when going from two to three lists. However, differences between using three to six lists are negligible. Furthermore, even for reasonably low heterogeneity levels ($a+b \approx 5$, meaning a standard deviation for the list inclusion probability of about $0.12$ when $L=3$), we find a relative bias of about $-50\%$. This means that estimates will be two times too small.

\begin{figure}[!ht]
    \centering
    \includegraphics{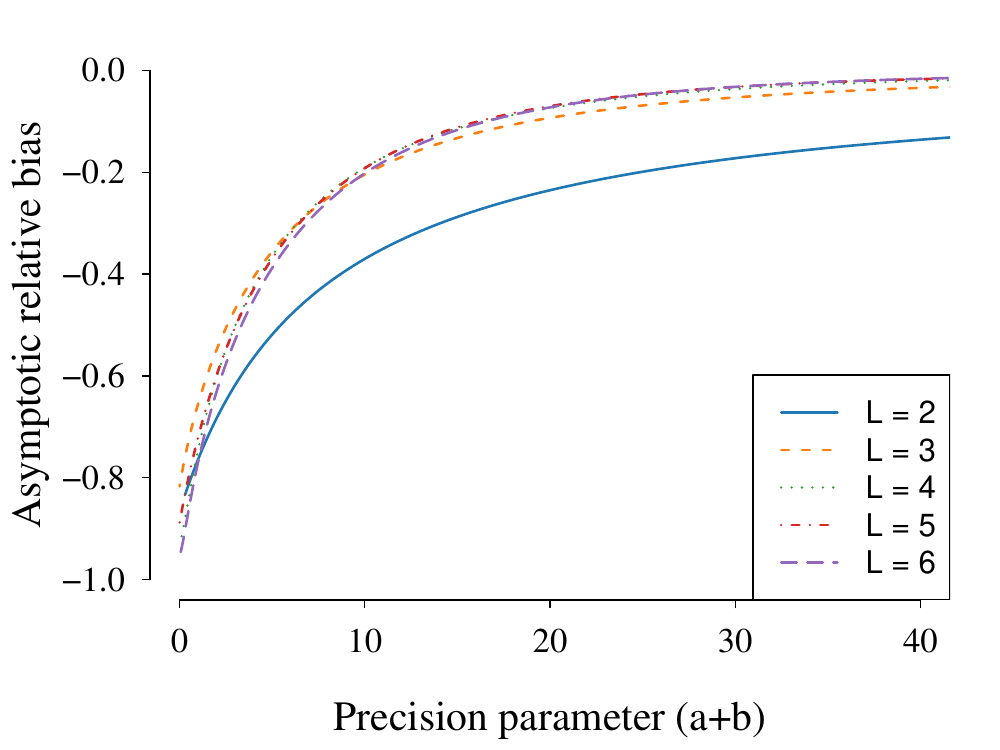}
    \caption{Asymptotic relative bias of population size estimators (assuming no highest order interaction) as a function of the precision parameter $a+b$ in the Beta model formulation and where we have fixed $p_{\bm{0}}\approx (b/(a+b))^L = 3/4$. 
    %{Even for reasonably low heterogeneity levels ($a+b \approx 5$, meaning a standard deviation for the list inclusion probability of about $0.12$ when $L=3$), we find a relative bias of about $-50\%$; estimates will be two times too small. Going from two to three lists reduces the effect of heterogeneity. However, the differences between the use of three to six lists are negligeable.}
    }
    \label{fig:beta_model_multi_list}
\end{figure}

\subsection{Summary}\label{sec:theory_summary}

{We have demonstrated how Theorem \ref{thm:characterization} can be used to evaluate the bias of estimates when the data has characteristics which break the assumption of no full-way interaction term. In particular, individual heterogeneity can result in substantial bias which can be either positive or negative. Furthermore, using more than three lists does not substantially reduce the consequence of individual heterogeneity under the Beta model which we considered.}

{The issue of individual heterogeneity has been addressed in some past studies through stratification or modeling. However, other issues remain even if heterogeneity can be accounted for. The most important might be observer effets. For example, an individual observed by one list might cause it to not appear on any other, or there might be systematic referal mechanisms between organizations. Theorem \ref{thm:characterization} provides an avenue to evaluate the consequences of these data characteristics on the bias (and ultimately accuracy) of estimates.}

\section{Visual Assessment of Robustness}\label{sec:visual_assessment}

Resampling techniques, including model-based and nonparametric bootstrapping approaches, provide powerful tools for the analysis of estimator properties in application to real data \citep{efron1982jackknife}. In particular, they may be used to estimate the bias of point estimates and the coverage of confidence intervals. However, such analyses rely on technical assumptions \citep{hall2013bootstrap}. They can be sensitive to modeling assumptions, and theoretical guarantees are only valid in large samples (including the need for large amounts of overlap data). In order to avoid any controversy of the kind found in \cite{Whitehead2019} and \cite{Vincent2020} regarding the setup of such experiments, we propose a simple visual assessment of estimate robustness and reliability. Our goal is that this visual assessment will be non-controversial and meaningful in practical applications. Our proposal is described in Section \ref{sec:visual_assessment_intro}. Section \ref{sec:visual_application} applies our tool to the aforementioned modern slavery datasets.

\subsection{Visualizing Estimate Trajectories}\label{sec:visual_assessment_intro}
{We propose to consider} the series of estimates which would have been obtained if the data had been collected sequentially. \ob{That is, we consider} series of estimates obtained as \ob{a function of the number of observed individuals. While this depends on the (unknown) order in which individuals have been observed, we may sample the order at random in order to obtain representative samples.} The series can \ob{also} be extended beyond the total size of the reported dataset through resampling. 

The behavior of the series may be indicative of convergence towards a stable estimate, or, if it is largely unstable, this can point to potential issues. \ob{Its behavior} can also be compared to what would be expected under a simple independence model fitted to the data, or under a more complex model fitted to the data. Differences between the behaviors of the series would then indicate a lack of fit of these simple or more complex models.

\ob{To be more precise, consider a dataset $\mathcal{D} = \{W_i\}_{i=1}^n$ of $n$ observations, where each $W_i \in \{0,1\}^L$ is an observed list inclusion pattern. From this dataset, we construct a series $\{Z_{i}\}_{i=1}^{n}$ which represents a hypothetical ordering of the observations in $\mathcal{D}$. This is obtained by choosing a permutation $\sigma$ of $\{1,2,\dots, n\}$ at random and setting $Z_i = W_{\sigma(i)}$. Furthermore, the series $Z_i$ is extended to $2n \geq i > n$ by choosing a second random permutation $\pi$ and setting $Z_{n+i} = W_{\pi(i)}$ for $1 \leq i \leq n$. This series $\{Z_i\}$ represents an hypothetical sample path of list inclusion pattern. The corresponding population size estimates, each computed using the first $n_{\text{obs}}$ data points $\{Z_i\}_{i=1}^{n_{\text{obs}}}$, are the \textit{estimate trajectories} which we focus on.}

%Each observation corresponds to an individual and these are ordered at random. For every integer $n_{\text{obs}}$ between $n/2$ and $n$, we compute the Independence, SparseMSE, dga, and LCMCR estimates corresponding to data from the first $n_{\text{obs}}$ individuals. Now, regarding the extension of the series, we consider a second random ordering of the $n$ individuals. Then, for integers $n_{\text{obs}}$ between $n$ and $2 n$, we compute these estimates for the data from the first $n$ individuals combined with data from the first $n_{\text{obs}} - n$ individuals from the second ordering. In this way, we obtain a series of estimates indexed by integers $n_{\text{obs}} \in (n/2, 2n)$. At point $n_{\text{obs}} = n$, these are the real data estimates. At point $n_{\text{obs}} = 2n$, these are the estimates on the data with doubled counts.

Figure \ref{fig:empirical_traj_example} shows such trajectories of dga estimates on data from an independence model fitted to the United Kingdom dataset.  That is, we have simulated a single dataset from the independence model and then applied {our proposed procedure to it} which provides different estimate trajectories for this data. Panel \textbf{A} shows a single trajectory with point estimates and confidence intervals. Panel \textbf{B} shows point estimate trajectories corresponding to 50 random orderings. \ob{The horizontal dotted line represents ground truth. The behaviour of these trajectories can be considered a best case scenario, given that the data came from a simple independence model. This can be compared with the application to real data in the following section.}

\begin{figure}[h!]
    \centering
    \includegraphics{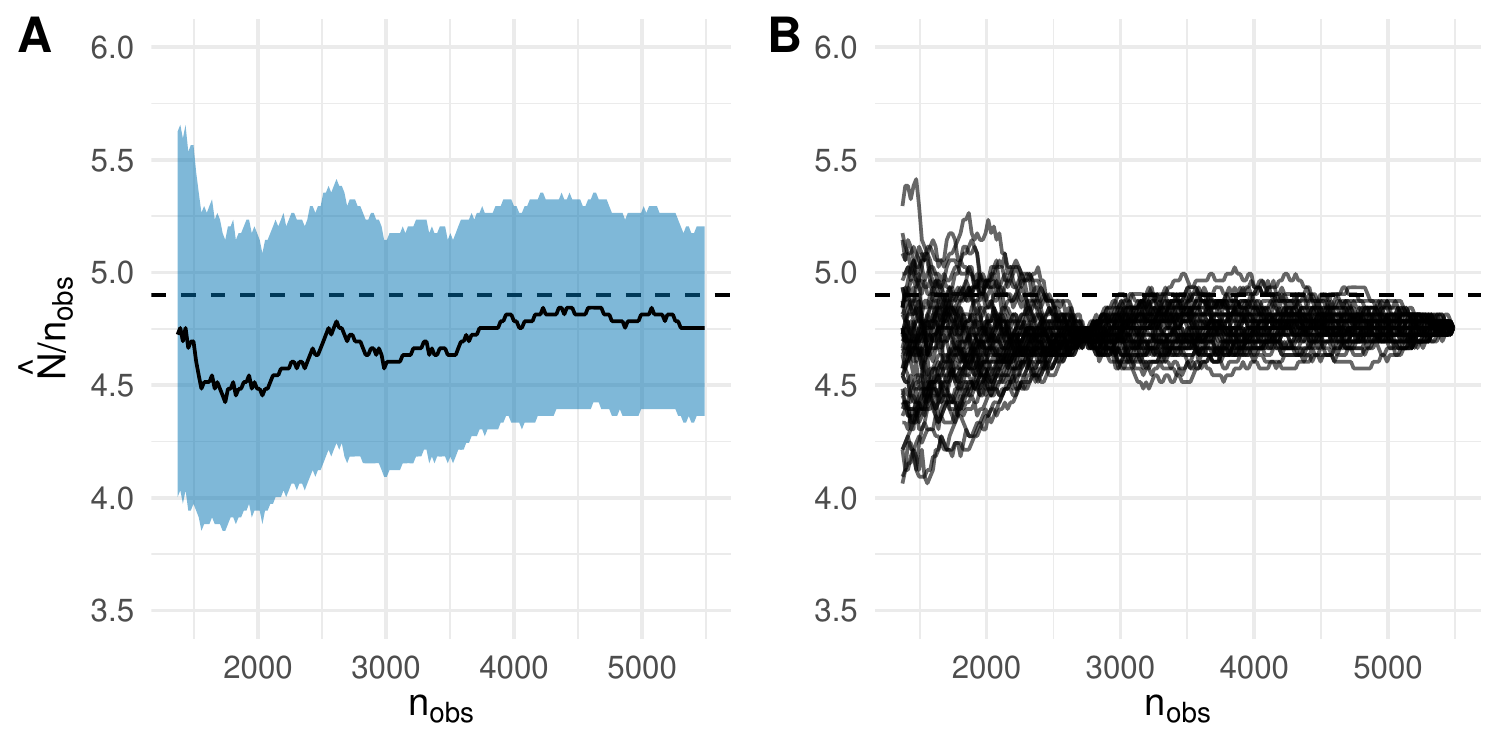}
    \caption{Trajectories of dga estimates (as the ratio of estimated population size to number of observations) on data simulated from and independence model fit to the United Kingdom dataset. The horizontal line represents ground truth population size. Panel \textbf{A} shows a single trajectory together with the $95\%$ credible intervals. Panel \textbf{B} shows $50$ random trajectories for the same dataset. 
    %{This figure illustrates the best case scenario of data from an independence model, using the dga estimator. It should be compared to Figure \ref{fig:empirical_traj_UK} and Figure \ref{fig:empirical_traj_Ned}.}
    }
    \label{fig:empirical_traj_example}
\end{figure}

%\textcolor{red}{Can you describe why this Figure is useful to a user.}

%{The figure illustrates the kind of behavior that can be expected in the \textit{best case scenario} of an independence model for the data and using the dga estimator. This should be compared with the application to real data in the following section.}

\subsection{Application to Real Data} \label{sec:visual_application}

We now present the result of our visualization {in application to the United Kingdom and Netherlands datasets. We focus on these two datasets because they are the largest and they are the ones for which sensitvity to individual observations is most noticeable.}

Figure \ref{fig:empirical_traj_UK} shows {trajectories of estimates} in application to the United Kingdom dataset. Figure \ref{fig:empirical_traj_Ned} shows trajectories of dga and SparseMSE estimates in application to the Netherlands dataset. {LCMCR and independence estimate trajectories have been omitted from figure \ref{fig:empirical_traj_Ned} since, like in the case of the United Kingdom data, they did not showcase high sensitivites.} The thin vertical lines indicate the number of observations at which the estimate trajectories coincides with real data estimates.

\begin{figure}[!h]
    \centering
    \includegraphics{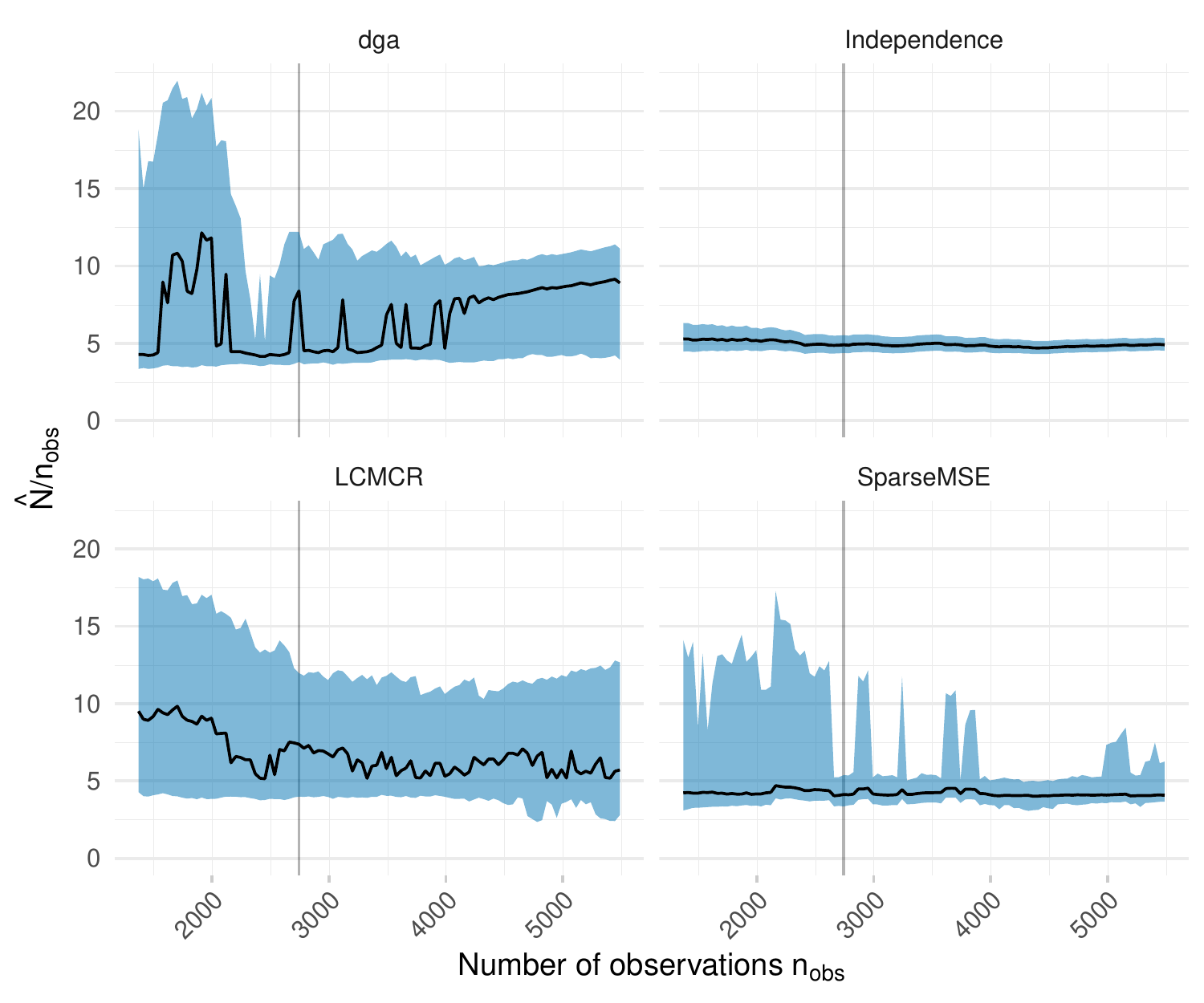}
    \caption{Visualization of estimate trajectories (as the ratio of estimated population size to number of observations) for the \textbf{United Kingdom} dataset. The horizontal line represents ground truth population size. The thin vertical lines indicate the number of observations in the United Kingdom dataset, at which the trajectory estimates coincide with real data estimates. 
    %{Notice the instabilities of the dga and SparseMSE estimates, when compared to LCMCR and Independence estimates. In particular, the SparseMSE confidence interval for the United Kingdom dataset appears to be too narrow.}
    }
    \label{fig:empirical_traj_UK}
\end{figure}

\begin{figure}[!h]
    \centering
    \includegraphics{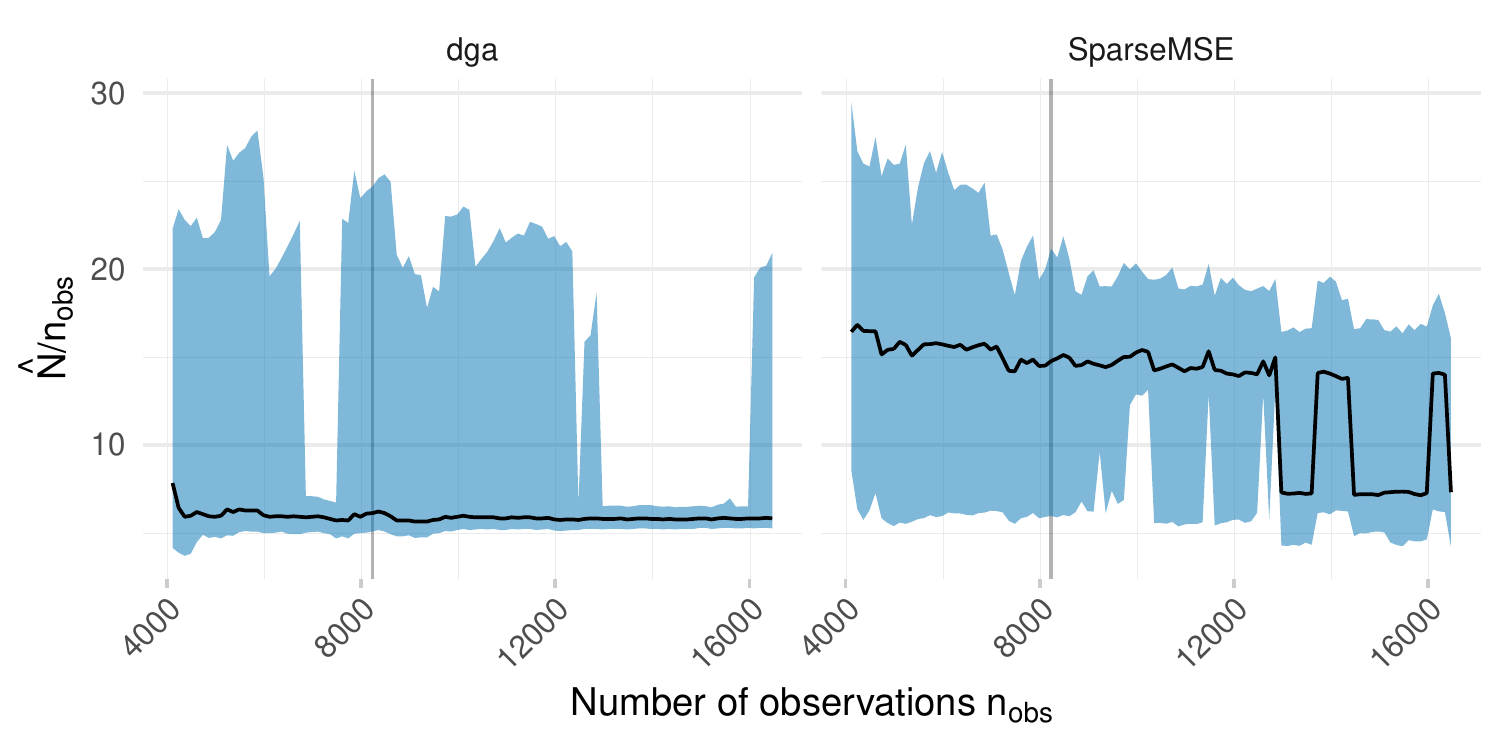}
    \caption{Visualization of estimate trajectories (as the ratio of estimated population size to number of observations) for the \textbf{Netherlands} dataset and for the dga and SparseMSE estimators. The thin vertical lines indicate the number of observations in the United Kingdom dataset, at which the trajectory estimates coincide with real data estimates. 
    %{While the real data estimates seem to properly capture uncertainty, estimates are still sensitive to the data in this case.}
    }
    \label{fig:empirical_traj_Ned}
\end{figure}

Looking at Figure \ref{fig:empirical_traj_UK}, the Independence and LCMCR estimates appear quite stable on the United Kingdom dataset. On the other hand, the dga and SparseMSE estimates are much more sensitive to individual observations. Regarding dga estimates, while $95\%$ credible intervals are relatively stable, the point estimates (posterior median) can significantly change as the result of observing only a few additional cases. Regarding SparseMSE estimates, the narrow confidence interval obtained on the United Kingdom dataset appears to be a fluke. With only a few less or a few more observations, the confidence interval becomes much larger, to between around 4 and 12 times the number of observed cases. This suggests that bootstrap confidence intervals used by SparseMSE may sometimes fail to properly account for uncertainty.

In the case of the Netherlands dataset shown in Figure \ref{fig:empirical_traj_Ned}, the behavior of the dga and sparseMSE estimates is quite different. The $95\%$ credible regions of dga now fluctuate much more, while SparseMSE point estimates are less stable in larger samples.

{In practice, the estimate trajectories can help diagnose lack of robustness. In the case of the SparseMSE estimate for the United Kingdom dataset, for instance, Figure \ref{fig:empirical_traj_UK} shows strong sensitivities and confidence intervals which are too narrow. This should be addressed by considering a broader range of plausible estimates.}

\clearpage
\newpage
\section{Discussion}\label{sec:discussion}
MSE has unique potential in helping assess the true scale of modern slavery. \ob{However, long-standing controversy in the literature} have come in the way of the broader implementation of MSE methodology.
%While MSE has the potential to address many applications of modern slavery, there are controversies in the literature that have prevented its broader implementation. 
As such, we address three major issues debated recently, namely, 
statistical aspects of MSE assumptions, robustness, and accuracy. First, we review the current state of the MSE literature, commonly used methods, and all publicly available modern slavery datasets. Next, we provide a reproducible analyses evaluating the accuracy of estimates, the consequences of MSE assumptions, and the robustness of estimates to small changes in data. Specifically, we utilize the internal consistency approach of \cite{Hook2012} to evaluate the accuracy of estimates when ground truth is available (section \ref{sec:internal_consistency_analysis}). Then, we assess the consequences of MSE assumptions  through a novel characterization of large sample bias (section \ref{sec:large-sample-bias}). Finally, we propose  a visual assessment of reliability and robustness (section \ref{sec:visual_assessment}). Our work highlights important practical and methodological challenges with MSE, which we summarize below. In addition, we comment on future  methodological research.

\paragraph{Practical Challenges and Recommendations}
Our work highlights some important statistical challenges practitioners face when using MSE to quantify modern slavery. First, we have shown in section \ref{sec:sensitivity_convergence} how estimates can be highly sensitive to the choice of tuning parameters. Due to this, we recommend for sensitivity analyses to be conducted, evaluated, and reported in all applications. Second, we have shown (section \ref{sec:visual_assessment}) that estimates can be highly sensitive to individual data points. Our proposed visualization of estimate trajectories can be used by practitioners as a diagnostic tool moving forward in analyses. Third, we have shown  (section \ref{sec:large-sample-bias}) that assumptions underlying MSE are highly influential. In particular, we have quantified how reasonable individual heterogeneity models can affect the bias of estimates from models used in practice. Thus, we recommend that users report, discuss, and justify underlying assumptions in their work moving forward. Ideally, enough information about data collection should be provided for these assumptions to be scrutinized. To summarize, this provides a set of simple guidelines and recommendations that can be used in MSE studies, complementing other recommendations that have been put forth \citep{Hook1999}.

\paragraph{Future Methodological Research}

The statistical challenges which we have highlighted point at a problem of \textit{underspecification} \citep{Damour2020underspecification}. In our context, this means that many seemingly reasonable approaches to MSE can give different results. There are two main reasons for this. First, there is the choice of underlying assumption of the form \textbf{A3} which we have discussed and shown to be highly influential. \cite{aleshin2021revisiting} has provided first steps towards emphasizing the role of these assumptions, but further work is needed to assess their suitability in real applications and to propose meaningful alternatives. Second, even if an assumption of the form \textbf{A3} can be justified, we have shown how estimates can be highly sensitive to tuning parameters and individual data points. That is, small errors or changes in data and arbitrary choices in model fitting procedures can lead to vastly different estimates. Methodology accounting for noise in the data, record linkage errors, and other sources of uncertainty, are needed to address these issues.

%Furthermore, we emphasize that have only scratched the surface of specifically \textit{statistical} issues involved with MSE. We have not considered the use of covariate information in MSE studies and we have limited our study of misspecified assumptions to the case of individual heterogeneity. Much more can also be considered regarding data collection practices, data security and privacy, data integration through record linkage, and the use of complementary information regarding data sources and their interdependencies. Additionally, there are issues of 

\paragraph{Software} Code and data for all analyses presented in this paper are available at \url{http://github.com/OlivierBinette/MSETools}

\paragraph{Appendix} The appendix contains summaries of the datasets considered in this paper and proofs which were omitted from the main text.

\clearpage
\newpage

\appendix

\bibliographystyle{chicago}

\begingroup
\raggedright
    \bibliography{MSE}
\endgroup

\section{Datasets Summary}\label{sec:appendix_data}

%{[TODO: write intro and complete this...]}

\begin{description}
    \item[United Kingdom] \cite{Silverman2014, Bales2015} considered data on potential victims of modern slavery collected as part of the National Crime Agency Strategic Assessment \citep{NCAOCCUKHTC2014}. This data identifies potential victims reported by six different sources of information in 2013: the police force (PF), the National Crime Agency (NCA), local authorities (LA), governmental organizations (GO), non-governmental organizations (NG), and the general public (GP). A total of 2744 cases are recorded, including 221 cases which appeared on more than one list. No cases appear on both the LA and GP lists, nor on both the LA and NCA lists. This may be due to the small size of the LA list (94 cases), to the way that the lists were constructed as part of the Strategic Assessment, or due to other unknown factors. No covariate information, such as the type of exploitation, country of origin, or demographic variables, is available for the potential victims. The mechanisms through which cases came to be observed on only one list or on more than one lists, such as the referral pathways between organizations, are not described.
    Using this data, \cite{Silverman2014, Bales2015} inferred a total of between 10,000 and 13,000 potential victims in the UK in 2013.
    
    \item[New Orleans] \cite{Bales2019} carried out a similar MSE study in the Greater New Orleans region of the United States. Eight (anonymous) organizations reported together a total of 185 ``confirmed cases of human trafficking'' for 2016, including 14 victims who appeared on more than one list. The authors noted challenges associated with this data, notably the definitional ambiguity of human trafficking cases and the differing goals of some organizations. One organization worked only with victims of labor trafficking, one only with victims of labor trafficking who were also victims of sex trafficking, another only reported victims of both sex trafficking and labor trafficking, and five only reported victims of sex trafficking. No covariate information is publicly available nor used in this study.
    Using this data, the authors inferred between around 650 and 1700 victims in 2016.
    
    \item[Netherlands] \cite{VanDijk2017} estimated the number of trafficked victims in the Netherlands disaggregated by sex, age, form of exploitation, nationality, and year between 2010 and 2015. This used data from a state-sponsored NGO, CoMensha, to which presumed cases of human trafficking are reported by the police and other organizations. The sources of information were grouped into six lists, the border police (K), the inspectorate (I), residential centers/shelters (O), national police (P), regional coordinators (R) and others (Z). A total of with 8,234 cases were observed, including 432 cases which appeared on more than one list. Only the aggregated data, with no covariate information, is publicly available. The authors noted challenges with this data, notably the definitional ambiguity of reported cases. In particular, the border police reported cases which do not correspond to the international definition of human trafficking. This lead the authors to produce estimates for both the inclusion and exclusion of this list. The authors also noted that the inclusion of covariates resulted in a significant change in estimated population size. Using covariates and all lists, they estimated around 42,000 victims. Using the same approach but ignoring covariates, they obtained an estimate of around 58,000 victims.
    
    \item[Western U.S.] \cite{Farrel2019} investigated the reporting and data collection mechanisms in place at law enforcement organizations regarding cases of human trafficking. They also carried out MSE studies at two locations, including at the ``Western site'' for which aggregated data is publicly available. This data records a total of 345 individuals which appeared in four lists, including a law enforcement list and three community service provider lists. A total of 23 individuals were observed in more than one list. The authors used covariate information with multiple imputation for missing data in order to obtain an estimate of between around 1,600 and 3,600 victims in the Western site in 2016.
    
    \item[Australia] \cite{Lyneham2019} reproduced the study of \cite{Silverman2014} with data collected by the Australian Federal Police. The dataset which appears in \cite{Lyneham2019} and which we have use consists of four lists and a total of 414 cases of victims observed between 2015-2016 and 2016-2017. The authors estimated between 900 to 1,500 total victims for this time period.
    
    \item[Other studies] Other MSE studies for the quantification of modern slavery have been carried out in Serbia, Ireland, Romania and Slovakia \citep{UNODC2018a, UNODC2018b, UNODC2018c, Vincent2020}. Data is not publicly available in these cases.
\end{description}

%These studies showcase some the challenges associated with the modern slavery data. First, aggregated data ignores covariate information and \textit{individual heterogeneity} is therefore to be expected regarding the probabilities of appearance in different lists. Second, definitional ambiguity and the focus of certain organizations on differing populations can additionally cause \textit{structural zeroes} in the data. Third, the data is very limited for some studies, with very few cases appearing on more than one lists. Fourth, \textit{referal pathways} between organizations are unknown. These issues are considered in more detail in section {[XX]}.

\section{Proof of Theorem \ref{thm:characterization}}\label{sec:proof_thm1}

%{[TODO: Rewrite proof using lemmas as Bekas suggested. Clean up.]}

The proof of the Theorem is separated in four parts. First we set up notation and the background of the problem. Steps 1-3 then provide the main ingredients necessary to establish \eqref{eq:thm_characterization}. While the proof is not complex, it does require a certain level of detail. Note that, throughout, convergence is meant to be in the almost sure sense.

\subsection*{Setup}
Let us make explicit the dependency of the counts $n_x$ on $N$ by writing
\begin{equation}\label{eq:def_counts_with_N}
    n_x^{(N)} = \# \{i \leq N : W_i = x\}, \quad x \in \{0,1\}^L \backslash\{\bm{0}\},\\
    n^{(N)} = \sum_{x \not = \bm{0}} n_x^{(N)},
\end{equation}
where $W_i$, $i=1, 2, 3, \dots$, is an independent sequence defined as in \eqref{eq:def_px}.

The count process $(n_x^{(N)})_{x \not = \bm{0}}$, for $N = 1,2,3,\dots$, is a Markov chain with distribution defined through the following:
\begin{enumerate}
    \item given $(n_x^{(N)})_{x \not = \bm{0}}$, with probability $p_{0}$ we have $(n_x^{(N+1)})_{x \not = \bm{0}} = (n_x^{(N)})_{x \not = \bm{0}}$;
    \item otherwise, for $x \in \{0,1\}^L\backslash\{\bm{0}\}$ distributed with probability mass function $q_x$, we have $n_x^{(N+1)} = n_x^{(N)} + 1$ and $n_{x'}^{(N+1)} = n_{x'}^{(N)}$ for every $x' \not = x$.
\end{enumerate}

Now recall that if $\gamma = 0$, then for any sequence $(n_x^{(N)})_{x \not = \bm{0}}$ distributed as above we have
\begin{equation}\label{eq:N_hat_consistency}
    \hat N\left( (n_x^{(N)})_{x \not = \bm{0}} \right)/N \rightarrow 1
\end{equation}
(almost surely) as $N \rightarrow \infty$, by assumption and definition of consistency.

Here, since we do not assume that $\gamma = 0$, our argument instead relies on the fact that there exists a random subsequence $N_k$, $k = 1,2,3,\dots$, such that $(n_x^{(N_k)})_{x \not = \bm{0}}$, $k = 1,2,3,\dots$ corresponds to the count process of a model with no full-way interaction and therefore
\begin{equation}
    \hat N((n_x^{(N_k)})_{x \not = \bm{0}})/k \rightarrow 1
\end{equation}
as $k \rightarrow \infty$. This is shown in Step 1 below. In Step 2, we show that the sequence $N_k$ satisfies $N_k / k \rightarrow (1-p_{\bm{0}}')/(1-p_{\bm{0}})$ almost surely as $k \rightarrow \infty$, allowing us to compute the limit of the ratio $\hat N((n_x^{(N_k)})_{x \not = \bm{0}})/N_k$. Step 3 shows that the relative asymptotic bias of $\hat N$ exists and we then easily deduce its form.

\subsection*{Step 1}

Here we define a random sequence $N_k$, $k = 1,2,3,\dots$, such that $\hat N((n_x^{(N_k)})_{x \not = \bm{0}})/k \rightarrow 1$ as $k \rightarrow \infty$. Let $p_{\bm{0}}' > 0$ be defined as
\begin{equation}\label{eq:def_px_prime}
    \frac{p_{\bm{0}}'}{1-p_{\bm{0}}'} = e^{\gamma} \frac{p_{\bm{0}}}{1-p_{\bm{0}}};\quad  p_{\bm{0}}' = \frac{p_{\bm{0}}e^\gamma}{1-p_{\bm{0}} + p_{\bm{0}}e^{\gamma}}
\end{equation}
and for $x \not = \bm{0}$ let $p_x' = (1-p_{\bm{0}}')q_x$. The probabilities $p_x'$, $x \in \{0,1\}^L$, define an alternative model to \eqref{eq:def_px} for which
\begin{equation}\label{eq:def_qx_prime}
    q_x' = \frac{p_x'}{(1-p_{\bm{0}}')} = q_x \quad \text{and} \quad \gamma' = \log \left( \frac{1-p_{\bm{0}}'}{p_{\bm{0}}'} \frac{q_{\text{odd}}'}{q_{\text{even}}'} \right) = 0.
\end{equation}
Now let us define the random sequence $N_k$ as follows as a function of $\{(n_x^{(N)})_{x \not = \bm{0}}\}_{N =1}^{\infty}$. Let $0 = s_0 < s_1 < s_2 < s_3 < \cdots$ be th sequence of integers such that $s_i$ is the smallest integer satisfying $n^{(s_i)} = i$. Let $t_1, t_2, t_3, \dots$ be an i.i.d. sequence of geometric random variables with mean $1/(1-p_{\bm{0}}')$, let $T_j = \sum_{i \leq j} t_i$, and let
\begin{equation}\label{eq:def_N_k}
    N_k = \sum_{i\,:\,T_i \leq k} (s_i - s_{i-1}).
\end{equation}
That is, $N_k$ is such that $N_k = s_j$ when $T_j \leq k < T_{j+1}-1$.

By construction, the distribution of $(n_{x}^{(N_{k+1})})_{x \not = \bm{0}}$ given $(n_{x}^{(N_{k})})_{x \not = \bm{0}}$ is the following:
\begin{enumerate}
    \item $n^{(N_{k+1})} = n^{(N_{k})}$ with probability $p_{\bm{0}}'$;
    \item otherwise $n_x^{(N_{k+1})} = n_x^{(N_{k})} + 1$ for some $x$ distributed following $q_x'$ and $n_{x'}^{(N_{k+1})} = n_{x'}^{(N_{k})}$ for $x' \not = x$.
\end{enumerate}
This is exactly the distribution of a count process obtained from the model with probabilities $p_x'$ defined in \eqref{eq:def_px_prime} and \eqref{eq:def_qx_prime}. These probabilities satisfy the no full-way interaction assumption \textbf{A3.1}, and it follows from the consistency of the estimator $\hat N$ that
\begin{equation}\label{eq:N_k_over_k_limit}
    \hat N\left((n_{x}^{(N_{k})})_{x \not = \bm{0}}\right)/k \rightarrow 1
\end{equation}
as $k \rightarrow \infty$.

\subsection*{Step 2} 

This step shows that $N_k/k \rightarrow (1-p_{\bm{0}}')/(1-p_{\bm{0}})$ as $k \rightarrow \infty$.
Note that the variables $s_i - s_{i-1}$ are independent with a geometric distribution of mean $1/(1-p_{\bm{0}})$ and they are independent from the variables $T_j$. By the strong law of large numbers, we have
\begin{equation}\label{eq:N_T_j_over_j_limit}
    N_{T_j}/j = \frac{1}{j} \sum_{i=1}^j (s_i - s_{i-1}) \rightarrow (1-p_{\bm{0}})^{-1}
\end{equation}
and 
\begin{equation}\label{eq:T_j_over_j_limit}
T_j / j = \frac{1}{j}\sum_{i=1}^j t_i  \rightarrow(1-p_{\bm{0}}')^{-1}
\end{equation}
as $j \rightarrow \infty$. Hence dividing the two we obtain $N_{T_j}/T_j \rightarrow (1-p_{\bm{0}}')/(1-p_{\bm{0}})$ as $j \rightarrow \infty$. Now for $j = j(k)$ such that $T_{j} \leq k < T_{j+1}$, we have $N_{k} = N_{T_j}$ and
\begin{equation}
N_k/k = \frac{N_{T_j}}{T_j}\frac{T_j}{k}\rightarrow (1-p_{\bm{0}}')/(1-p_{\bm{0}})
\end{equation}
follows from the fact that $T_j/k \rightarrow 1$ as $k \rightarrow \infty$.

\subsection*{Step 3}

Combining steps 1 and 2, we obtain that $\hat N((n_x^{(N_k)})_{x \not = \bm{0}})/N_k \rightarrow(1-p_{\bm{0}})/(1-p_{\bm{0}}')$ as $k \rightarrow \infty$. It can easily be verified that $(n_x^{(N)})_{x \not = \bm{0}}$ is also a subsequence of $(n_x^{(N_k)})_{x \not = \bm{0}}$, and therefore $\hat N/N$ converges with
$$
    \hat N / N = \hat N((n_x^{(N)})_{x \not = \bm{0}})/N \rightarrow(1-p_{\bm{0}})/(1-p_{\bm{0}}')
$$
Using expression \eqref{eq:def_px_prime} for $p_{\bm{0}}'$ and simplifying, we obtain \eqref{eq:thm_characterization}. This concludes the proof.

\section{Proof of Proposition \ref{prop:1}}\label{sec:proof_proposition}

Suppose there exists a consistent population size estimator $\hat N$ for a model for which no function $f$ satisfies the condition of Proposition \ref{prop:1}. That is, there exists two sets of probabilities $p_x$ and $\tilde p_x$, $x \in \{0,1\}^L$, such that $p_{\bm{0}} \not = \tilde p_{\bm{0}}$ while $q_x = p_x/(1-p_{\bm{0}}) = \tilde p_x/(1-\tilde p_{\bm{0}}) = \tilde q_x$ for $x \not = \bm{0}$. 

Continuing with similar notations as in the proof of Theorem \ref{thm:characterization}, let $n_x^{(N)}$ and $\tilde n_x^{(N)}$, $x \not = \bm{0}$, be the two sequences of observed counts corresponding to the model probabilities $p_x$ and $p_x'$, respectively. Let $n^{(N)} = \sum_{x \not = \bm{0}} n_x^{(N)}$, $\tilde n^{(N)} = \sum_{x \not = \bm{0}} \tilde n_x^{(N)}$, and let $s_k$ and $\tilde s_k$ be sequences of the smallest integers satisfying $n^{(s_k)} = k$ and $\tilde n^{(\tilde s_k)} = k$.

Since $\hat N( (n_x)^{(s_k)}_{x \not = \bm{0}} )$ is a subsequence of $\hat N( (n_x)^{(N)}_{x \not = \bm{0}} )$ and by consistency of $\hat N$, we have that $\hat N( (n_x)^{s_k}_{x \not = \bm{0}} )/ s_k \rightarrow 1$ almost surely as $k \rightarrow \infty$. For the same reasons, $\hat N( (\tilde n_x)^{(\tilde s_k)}_{x \not = \bm{0}} )/ \tilde s_k \rightarrow 1$ almost surely as $k \rightarrow \infty$.

Now notice that the two sequences $(n_x)^{s_k}_{x \not = \bm{0}}$ and $(\tilde n_x)^{\tilde s_k}_{x \not = \bm{0}}$ have exactly the same distribution since $q_x = \tilde q_x$, $x \not = \bm{0}$. Therefore $\hat N ((n_x)^{(s_k)}_{x \not = \bm{0}})/ \hat N( (\tilde n_x)^{\tilde (s_k)}_{x \not = \bm{0}} ) \rightarrow 1$ as $k \rightarrow \infty$.

It follows from the above that $s_k / \tilde s_k \rightarrow 1$. However, this can only happen if $p_{\bm{0}} = \tilde p_{\bm{0}}$, which is not the case here. Indeed, the increments $s_{k+1} - s_{k}$ are independent with geometric distribution of mean $(1-p_{\bm{0}})^{-1}$, and by the law of large numbers it follows that $s_k/k \rightarrow (1-p_{\bm{0}})^{-1}$. Similarly, $\tilde s_k/k \rightarrow (1-\tilde p_{\bm{0}})^{-1}$, from which it follows that $s_k/\tilde s_k \rightarrow (1-\tilde p_{\bm{0}})/(1-p_{\bm{0}}) \not = 1$.

\end{document}